\documentclass{tlp}

\usepackage{amsmath}
\usepackage{amssymb}
\usepackage{latexsym}
\usepackage{stmaryrd}
\usepackage{hyperref}

\hypersetup{breaklinks}

\newcommand{\mysim}{\sim\!}

\newtheorem{lemma}{Lemma}
\newtheorem{proposition}[lemma]{Proposition}
\newtheorem{theorem}{Theorem}
\newtheorem{definition}{Definition}
\newtheorem{example}{Example}
\newtheorem{remark}{Remark}
\newtheorem{corollary}{Corollary}

\newcommand{\lsem}{\mbox{$\lbrack\!\lbrack$}}
\newcommand{\rsem}{\mbox{$\rbrack\!\rbrack$}}
\newcommand{\ctrue}{\mathsf{true}}
\newcommand{\cfalse}{\mathsf{false}}

\newcommand{\mnot}{\sim\!\!}

\newcommand{\per}{\mbox{{\tt .}}}
\newcommand{\lpa}{\mbox{{\tt (}}}
\newcommand{\rpa}{\mbox{{\tt )}}}

\newcommand{\C}{\mathcal{C}}
\newcommand{\eval}{\mathsf{eval}}

\title[Minimum Model Semantics for Extensional Higher-order LP with Negation]{Minimum Model Semantics for Extensional Higher-order Logic Programming with Negation\thanks{This research is being supported by the Greek General
Secretariat for Research and Technology, the National Development
Agency of Hungary, and the European Commission (European Regional
Development Fund) under a Greek-Hungarian intergovernmental programme of
Scientific and Technological collaboration. Project
title: ``Extensions and Applications of Fixed Point Theory for
Non-Monotonic Formalisms''. It is also supported
by grant no. ANN 110883 from the National Foundation of Hungary for
Scientific Research.}}
\author[A. Charalambidis, Z. \'{E}sik and P. Rondogiannis]{%
        Angelos Charalambidis \\
        Department of Informatics \& Telecommunications,
        University of Athens, Greece\\
        \email{a.charalambidis@di.uoa.gr}\and
        Zolt\'{a}n \'{E}sik \\
        Department of Computer Science,
        University of Szeged, Hungary \\
        \email{ze@inf.u-szeged.hu} \and
        Panos Rondogiannis \\
        Department of Informatics \& Telecommunications,
        University of Athens, Greece\\
        \email{prondo@di.uoa.gr}}

\usepackage{comment}

\begin{document}

\maketitle

\begin{abstract}
Extensional higher-order logic programming has been introduced as a generalization
of classical logic programming. An important characteristic of this paradigm is that it preserves
all the well-known properties of traditional logic programming. In this paper we consider the semantics
of negation in the context of the new paradigm. Using some recent results from non-monotonic fixed-point
theory, we demonstrate that every higher-order logic program with negation has a unique {\em minimum}
infinite-valued model. In this way we obtain the first purely model-theoretic semantics for negation
in extensional higher-order logic programming. Using our approach, we resolve an old paradox that was
introduced by W. W. Wadge in order to demonstrate the semantic difficulties of higher-order logic programming.
\end{abstract}

\section{Introduction}
Extensional higher-order logic programming has been proposed~\cite{Wadge91,CHRW10,CHRW13}
as a generalization of classical logic programming. The key idea behind this paradigm is
that all predicates defined in a program denote sets and therefore one can use standard extensional
set theory in order to understand their meaning and to reason about them. For example, consider the
following simple extensional higher-order program~\cite{CHRW13} stating that a band (musical
ensemble) is a group that has at least a singer and a guitarist:
\[
\begin{array}{l}
\mbox{\tt
band(B):-singer(S),B(S),guitarist(G),B(G).}
\end{array}
\]
Suppose that we also have a database of musicians:
\[
\begin{array}{l}
\mbox{\tt singer(sally).}\\
\mbox{\tt singer(steve).}\\
\mbox{\tt guitarist(george).}\\
\mbox{\tt guitarist(grace).}
\end{array}
\]
We can then ask the query
\(
\mbox{\tt ?-band(B)}
\).
Since predicates denote sets, an extensional higher-order language will return answers such
as ${\tt B} =\{{\tt sally,george}\} \cup {\tt L}$, having the meaning that every set that
contains at least {\tt sally} and {\tt george} is a potential band.

A consequence of the set-theoretic nature of extensional higher-order logic programming is the fact
that its semantics and its proof theory smoothly extend the corresponding ones for traditional
(ie., first-order) logic programming. In particular, every program has a unique minimum Herbrand
model which is the greatest lower bound of all Herbrand models of the program and the least
fixed-point of an immediate consequence operator associated with the program; moreover, there exists an SLD
resolution proof-procedure which is sound and complete with respect to the minimum model semantics.

One basic property of all the higher-order predicates that can be defined in the language of~\cite{CHRW13}
is that they are {\em monotonic}. Intuitively, the monotonicity property states that if a
predicate is true of a relation {\tt R} then it is also true of every superset of
{\tt R}. In the above example, it is clear that if {\tt band} is true of a relation {\tt B} then
it is also true of any band that is a superset of {\tt B}. However, there are many natural higher-order
predicates that are {\em non-monotonic}. Consider for example a predicate {\tt single\_singer\_band}
which (apparently) defines a band that has a unique singer:
\[
\begin{array}{l}
\mbox{\tt single\_singer\_band(B):-band(B),not two\_singers(B).}\\
\mbox{\tt two\_singers(B):-B(S1),B(S2),singer(S1),singer(S2),not(S1=S2).}
\end{array}
\]

The predicate {\tt single\_singer\_band} is obviously non-monotonic since it is satisfied
by the set $\{{\tt sally,george}\}$ but not by the set $\{{\tt sally,steve,george}\}$.
In other words, the semantics of~\cite{CHRW13} is not applicable to this extended
higher-order language. We are therefore facing the same problem that researchers faced
more than twenty years ago when they attempted to provide a sensible semantics to classical
logic programs with negation; the only difference is that the problem now reappears in
a much more general context, namely in the context of higher-order logic programming.

The solution we adopt is relatively simple to state (but non-trivial to materialize): it suffices
to generalize the well-founded construction~\cite{GelderRS91,Przymusinski89} to higher-order programs.
For this purpose, we have found convenient to use a relatively recent logical characterization
of the well-founded semantics through an infinite-valued logic~\cite{tocl05} and also the
recent abstract fixed-point theory for non-monotonic functions developed in~\cite{pls2013,tocl14}. This brings
us to the two main contributions of the present paper:
\begin{itemize}
\item We provide the first model-theoretic semantics for extensional higher-order
      logic programming with negation. In this way we initiate the study of a non-monotonic
      formalism that is much broader than classical logic programming with negation.

\item We provide further evidence that extensional higher-order logic programming is a
      natural generalization of classical logic programming, by showing that all the well-known
      properties of the latter also hold for the new paradigm.
\end{itemize}

In the next section we provide an introduction to the proposed semantics for
higher-order logic programming and the remaining sections provide the formal development
of this semantics. The proofs of all the results have been moved to corresponding appendices.

\section{An Intuitive Overview of the Proposed Semantics}\label{section2}
The starting point for the semantics proposed in this paper is the
{\em infinite-valued semantics} for ordinary logic programs with negation,
as introduced in~\cite{tocl05}. In this section we give an intuitive introduction
to the infinite-valued approach and discuss how it can be extended to the
higher-order case.

The infinite-valued approach was introduced in order to provide a {\em minimum model}
semantics to logic programs with negation. As we are going to see shortly,
it is compatible with the well-founded semantics but it is purely model-theoretic\footnote{In
the same way that the equilibrium logic approach of~\cite{Pearce96} gives a purely logical reconstruction
of the stable model semantics.}. The main idea of this approach can be explained with a simple example. 
Consider the program:
\[
\begin{array}{lll}
 {\tt p} & \leftarrow & \\
 {\tt r} & \leftarrow & \mysim {\tt p}\\
 {\tt s} & \leftarrow & \mysim {\tt q}
\end{array}
\]
Under the well-founded semantics both {\tt p} and {\tt s} receive the value {\em True}.
However, {\tt p} is in some sense ``truer'' than {\tt s}. Namely, {\tt p}
is true because there is a rule which says so, whereas {\tt s} is true only because we
are never obliged to make {\tt q} true. In a sense, {\tt s} is true only by default. This
gave the idea of adding a ``default'' truth value $T_1$ just below the ``real'' truth $T_0$,
and (by symmetry) a weaker false value $F_1$ just above (``not as false as'') the real false
$F_0$. We can then understand negation-as-failure as combining ordinary negation with a weakening.
Thus $\mysim F_0 = T_1$ and $\mysim T_0 = F_1$. Since negations can effectively be iterated, the
infinite-valued approach requires a whole  sequence $\ldots,T_3, T_2, T_1$ of weaker and weaker truth values
below $T_0$ but above the neutral value $0$; and a mirror image sequence $F_1, F_2,F_3,\ldots$
above $F_0$ and below $0$. In fact, to capture the well-founded model
in full generality, we need a $T_\alpha$ and a $F_\alpha$ for every countable ordinal $\alpha$.
In other words, the underlying truth domain of the infinite-valued approach is:
\[
F_0 < F_1 <\!\cdots\!< F_\omega <\!\cdots\!< F_\alpha < \!\cdots\! < 0 <\!\cdots\!< T_\alpha <\!\cdots\!< T_\omega <\!\cdots\!< T_1 < T_0
\]
As shown in~\cite{tocl05}, every logic program $\mathsf{P}$ with negation
has a unique {\em minimum} infinite-valued model $M_{\mathsf{P}}$. Notice that
$M_\mathsf{P}$ is minimum with respect to a relation $\sqsubseteq$ which compares
interpretations in a stage-by-stage manner (see~\cite{tocl05} for details).
As it is proven in~\cite{tocl05}, if we collapse all the
$T_\alpha$ and $F_\alpha$ to {\em True} and {\em False} respectively,
we get the well-founded model. For the example program above, the minimum model is
$\{({\tt p},T_0),({\tt q},F_0),({\tt r},F_1),({\tt s},T_1)\}$. This collapses to
$\{({\tt p},\mbox{\em True}),({\tt q},\mbox{\em False}),({\tt r},\mbox{\em False}),({\tt s},\mbox{\em True})\}$,
which is the well-founded model of the program.

As shown in~\cite{tocl05}, one can compute the minimum infinite-valued
model as the least fixed point of an operator $T_\mathsf{P}$. It can easily be seen
that $T_\mathsf{P}$ is {\em not} monotonic with respect to the ordering relation $\sqsubseteq$
and therefore one can not obtain the least fixed point using the classical Knaster-Tarski
theorem. However, $T_\mathsf{P}$ possesses some form of {\em partial monotonicity}. More specifically,
as it is shown in~\cite{tocl05,tocl14},
$T_\mathsf{P}$ is $\alpha$-monotonic for all countable ordinals $\alpha$,
a property that guarantees the existence of the least fixed point. Loosely speaking,
the property of $T_\mathsf{P}$ being $\alpha$-monotonic means that the operator
is monotonic when we restrict attention to interpretations
that are equal for all levels of truth values that are less than $\alpha$. In other words,
$T_\mathsf{P}$ is monotonic in stages (but not overall monotonic).

The $T_\mathsf{P}$ operator is a higher-order function since it takes as argument an
interpretation and returns an interpretation as
the result. This observation leads us to the main concept that helps us extend the infinite-valued
semantics to the higher-order case. The key idea is to demonstrate that the denotation of
every expression of predicate type in our higher-order language, is $\alpha$-monotonic
for all ordinals $\alpha$ (see Lemma~\ref{monotonicity-of-semantics}). This
property ensures that the immediate consequence operator of every program is also $\alpha$-monotonic
for all $\alpha$ (see Lemma~\ref{tp-a-monotonic}), and therefore it has a least
fixed-point which is a model of the program. Actually, this same model can also be
obtained as the greatest lower bound of all the Herbrand models of the program
(see Theorem~\ref{model-intersection}, the {\em model intersection theorem}).
In other words, the semantics of extensional higher-order logic programming
with negation preserves all the familiar properties of classical logic programming
and can therefore be considered as a natural generalization of the latter.

\section{Non-Monotonic Fixed Point Theory}\label{section3}
The main results of the paper will be obtained using some recent results
from non-monotonic fixed point theory~\cite{pls2013,tocl14}.
The key objective of this area of research is to obtain novel fixed point
results regarding functions that are not necessarily monotonic. In particular,
the results obtained in~\cite{pls2013,tocl14} generalize the classical
results of monotonic fixed-point theory (namely Kleene's theorem and also
the Knaster-Tarski theorem). In this section we provide the necessary material
from~\cite{pls2013,tocl14} that will be needed in the next sections.

Suppose that $(L,\leq)$ is a complete lattice in which the least upper bound
operation is denoted by $\bigvee$ and the least element is denoted by $\perp$.
Let $\kappa >0$ be a fixed ordinal. We assume that for each ordinal $\alpha<\kappa$,
there exists a preordering $\sqsubseteq_\alpha$ on $L$. We write $x =_\alpha y$ iff
$x \sqsubseteq_\alpha y$ and $y \sqsubseteq_\alpha x$. We define $x \sqsubset_\alpha y$
iff $x \sqsubseteq_\alpha y$ but $x =_\alpha y$ does not hold. Moreover, we write
$x \sqsubset y$ iff $x \sqsubset_\alpha y$ for some $\alpha < \kappa$.
Finally, we define $x \sqsubseteq y$ iff $x \sqsubset y$ or $x = y$.

Let $x \in L$ and $\alpha < \kappa$. We define
$(x]_\alpha = \{ y : \forall \beta < \alpha\ x =_\beta y \}$.

A key property that will be used throughout the paper is that if the above preordering relations
satisfy certain simple axioms, then the structure $(L,\sqsubseteq)$ is a
complete lattice; moreover, every function $f:L\rightarrow L$ that satisfies some
restricted form of monotonicity, has a least fixed point. These ideas
are formalized by the following definitions and results.
\begin{definition}\label{def-basic model}
Let $(L,\leq)$ be a complete lattice equipped with preorderings $\sqsubseteq_\alpha$
for all $\alpha < \kappa$. Then, $L$ will be called a {\em basic model} if and only
if it satisfies the following axioms:
\begin{enumerate}
\item \label{axiom1}
For all $x,y\in L$ and all $\alpha < \beta < \kappa$, if $x \sqsubseteq_\beta y$ then $x =_\alpha y$.

\item \label{axiom2}
For all $x,y\in L$, if $x =_\alpha y$ for all $\alpha < \kappa$ then $x = y$.

\item \label{axiom3}
Let $x \in L$ and $\alpha < \kappa$. Let $X \subseteq (x]_\alpha$. Then, there exists
$y$ (denoted by $\bigsqcup_\alpha X$) such that $X \sqsubseteq_\alpha y$\footnote{We write
$X \sqsubseteq_\alpha y$ iff forall $x \in X$ it holds $x \sqsubseteq_\alpha y$.}
and for all
$z \in (x]_\alpha$ such that $X \sqsubseteq_\alpha z$, it holds $y \sqsubseteq_\alpha z$
and $y \leq z$.

\item  \label{axiom6}
If $x_j, y_j \in L$ and $x_j \sqsubseteq_\alpha y_j$ for all $j \in J$ then
$\bigvee\{x_j : j \in J\} \sqsubseteq_\alpha \bigvee\{ y_j: j \in J \}$.


\end{enumerate}
\end{definition}

\begin{lemma}\label{basic-model-lattice}
Let $L$ be a basic model. Then, $(L, \sqsubseteq)$ is a complete lattice.
\end{lemma}

\begin{definition}
Let $A,B$ be basic models and let $\alpha < \kappa$.
A function $f : A \rightarrow B$ is called $\alpha$-monotonic if for all
$x, y \in A$ if $x \sqsubseteq_\alpha y$ then $f(x) \sqsubseteq_\alpha f(y)$.
\end{definition}

It should be noted that even if a function $f$ is $\alpha$-monotonic for all $\alpha < \kappa$, then
it need not be necessarily monotonic with respect to the relation $\sqsubseteq$ (for a
counterexample, see~\cite[Example 5.7, pages 453--454]{tocl05}). Therefore, the standard
tools of classical fixed point theory (such as the Knaster-Tarski theorem), do not suffice
in order to find the least fixed point of $f$ with respect to the relation $\sqsubseteq$.

Let us denote by $[A  \stackrel{m}{\rightarrow} B]$
the set of functions from $A$ to $B$ that are $\alpha$-monotonic
for all $\alpha < \kappa$.

\begin{theorem}\label{a-monotonic-has-least-fixpoint}
Let $L$ be a basic model and assume that $f\in [L  \stackrel{m}{\rightarrow} L]$.
Then, $f$ has a $\sqsubseteq$-least pre-fixed point,
which is also the $\sqsubseteq$-least fixed point of $f$.
\end{theorem}

The above theorem will be our main tool for establishing the fact that the immediate
consequence operator of any extensional higher order logic program, always has a least
fixed point, which is a model of the program.

\section{The Syntax of the Higher-Order Language ${\cal H}$}\label{section4}
In this section we introduce the higher-order language ${\cal H}$,
which extends classical first-order logic programming to a
higher-order setting. The language ${\cal H}$ is based on a simple
type system that supports two base types: $o$, the boolean domain,
and $\iota$, the domain of individuals (data objects). The composite
types are partitioned into three classes: functional (assigned to
individual constants, individual variables and function symbols),
predicate (assigned to predicate constants and variables) and argument
(assigned to parameters of predicates).
\begin{definition}
A type can either be functional, predicate, argument, denoted by
$\sigma$, $\pi$ and $\rho$ respectively and defined as:
\begin{align*}
\sigma & :=  \iota \mid \iota \rightarrow \sigma  \\
\pi   & := o \mid \rho \rightarrow \pi  \\
\rho & :=  \iota \mid \pi
\end{align*}
We will use $\tau$ to denote an arbitrary type (either functional, predicate or argument
one).
\end{definition}

As usual, the binary operator $\rightarrow$ is right-associative. A
functional type that is different than $\iota$ will often be written
in the form $\iota^n \rightarrow \iota$, $n\geq 1$ (which stands for
$\iota \rightarrow \iota \rightarrow \cdots \rightarrow \iota$ $(n+1)$-times).
Moreover, it can be easily seen that every predicate type $\pi$ can be written uniquely
in the form $\rho_1 \rightarrow \cdots \rightarrow \rho_n \rightarrow o$,
$n\geq 0$ (for $n=0$ we assume that $\pi=o$).
We can now proceed to the definition of ${\cal H}$, starting from
its alphabet and continuing with expressions and program clauses:
\begin{definition}
The \emph{alphabet} of the higher-order language ${\cal H}$ consists
of the following:
\begin{enumerate}
\item {\em Predicate variables} of every predicate type $\pi$
      (denoted by capital letters such as
      $\mathsf{P,Q,R,\ldots}$).

\item {\em Predicate constants} of every predicate type $\pi$
      (denoted by lowercase letters such as
      $\mathsf{p,q,r,\ldots}$).

\item {\em Individual variables} of type $\iota$
      (denoted by capital letters such as
      $\mathsf{X,Y,Z,\ldots}$).

\item {\em Individual constants} of type $\iota$ (denoted by lowercase
      letters such as $\mathsf{a,b,c,\ldots}$).

\item {\em Function symbols} of every functional type $\sigma \neq \iota$
      (denoted by lowercase letters such as $\mathsf{f,g,h,\ldots}$).

\item The following {\em logical constant symbols}: the constants
      $\cfalse$ and $\ctrue$ of type $o$; the equality  constant $\approx$
      of type $\iota \rightarrow \iota \rightarrow o$; the generalized disjunction
      and conjunction constants $\bigvee_{\pi}$ and $\bigwedge_{\pi}$ of type
      $\pi \rightarrow \pi \rightarrow \pi$, for every predicate type $\pi$;
      the generalized inverse implication constants $\leftarrow_{\pi}$, of type
      $\pi \rightarrow \pi \rightarrow o$, for every predicate type
      $\pi$; the existential quantifier $\exists_{\rho}$, of type
      $(\rho \rightarrow o)\rightarrow o$, for every argument type
      $\rho$; the negation constant $\mnot\,$ of type $o \rightarrow o$.

\item The {\em abstractor} $\lambda$ and the parentheses ``$\mathsf{(}$'' and ``$\mathsf{)}$''.
\end{enumerate}
The set consisting of the predicate variables and the individual
variables of ${\cal H}$ will be called the set of {\em argument
variables} of ${\cal H}$. Argument variables will be usually denoted
by $\mathsf{V}$ and its subscripted versions.
\end{definition}
\begin{definition}

The set of {\em expressions} of the higher-order language ${\cal H}$
is defined as follows:
\begin{enumerate}
\item Every predicate variable (respectively, predicate constant)
      of type $\pi$ is an expression of type $\pi$; every
      individual variable (respectively, individual constant) of
      type $\iota$ is an expression of type $\iota$; the
      propositional constants $\cfalse$ and $\ctrue$ are
      expressions of type $o$.

\item If $\mathsf{f}$ is an $n$-ary function symbol and $\mathsf{E}_1,
      \ldots, \mathsf{E}_n$ are expressions of type $\iota$,
      then $(\mathsf{f}\,\,\mathsf{E}_1 \cdots \mathsf{E}_n)$ is an
      expression of type $\iota$.

\item If $\mathsf{E}_1$ is an expression of type $\rho \rightarrow
      \pi$ and $\mathsf{E}_2$ is an expression of type $\rho$, then
      $(\mathsf{E}_1\ \mathsf{E}_2)$ is an expression of type $\pi$.

\item If $\mathsf{V}$ is an argument variable of type $\rho$ and
      $\mathsf{E}$ is an expression of type $\pi$, then
      $(\lambda\mathsf{V}.\mathsf{E})$ is an expression of type
      $\rho \rightarrow \pi$.

\item If $\mathsf{E}_1,\mathsf{E}_2$ are expressions of type $\pi$,
      then  $(\mathsf{E}_1 \bigwedge_{\pi} \mathsf{E}_2)$ and
      $(\mathsf{E}_1 \bigvee_{\pi} \mathsf{E}_2)$ are expressions
      of type $\pi$.

\item If $\mathsf{E}$ is an expression of type $o$, then
      $(\mnot \mathsf{E})$ is an expression of type $o$.

\item If $\mathsf{E}_1,\mathsf{E}_2$ are expressions of type $\iota$,
      then $(\mathsf{E}_1 \approx \mathsf{E}_2)$ is an expression
      of type $o$.

\item If $\mathsf{E}$ is an expression of type $o$ and $\mathsf{V}$ is
      a variable of type $\rho$ then $(\exists_{\rho}
      \mathsf{V}\,\mathsf{E})$ is an expression of type $o$.
\end{enumerate}
\end{definition}

To denote that an expression $\mathsf{E}$ has type $\tau$ we will write
$\mathsf{E} : \tau$.
The notions of \emph{free} and \emph{bound} variables of an
expression are defined as usual. An expression is called
\emph{closed} if it does not contain any free variables.

\begin{definition}
A {\em program clause} is a clause $\mathsf{p} \leftarrow_\pi \mathsf{E}$  where
$\mathsf{p}$ is a predicate constant of type $\pi$ and $\mathsf{E}$
is a closed expression of type $\pi$.
A {\em program} is a finite set of program clauses.
\end{definition}
\begin{example}\label{subset-predicate-example}
The {\tt subset} predicate can be defined in ${\cal H}$ as follows:
%
%
%
%
\begin{eqnarray*}
\mathtt{subset} \leftarrow_{\pi\rightarrow\pi\rightarrow o}
    \lambda\mathtt{P}\per \lambda\mathtt{Q}.
        \mnot\exists\mathtt{X}\lpa\lpa\mathtt{P}\ \mathtt{X}\rpa \wedge
                               \mnot\lpa\mathtt{Q}\ \mathtt{X}\rpa\rpa
\end{eqnarray*}
The $\mathtt{subset}$ predicate is defined by a $\lambda$-expression (which obviates
the need to have the formal parameters of the predicate in the left-hand side
of the definition). Moreover, in the right-hand side we have an explicit existential
quantifier for the variable $\mathtt{X}$ (in Prolog, if a variable appears in the
body of a clause but not in the head, then it is implicitly existentially quantified).
\end{example}

\section{The Semantics of the Higher-Order Language ${\cal H}$}\label{section5}
In this section we specify the semantics of ${\cal H}$. We start with the semantics
of types and proceed to the semantics of expressions.

The meaning of the boolean type $o$ is equal to a partially ordered set
$(V,\leq)$ of truth values. The number of truth values of $V$ will be specified
with respect to an ordinal $\kappa>0$. All the results of the paper hold for {\em every}
initial selection of $\kappa$. The set $(V,\leq)$ is therefore
\[
F_0 < F_1 <\!\cdots\!< F_\alpha < \!\cdots\! < 0 <\!\cdots\!< T_\alpha <\!\cdots\!< T_1 < T_0
\]
where $\alpha < \kappa$.

\begin{definition}
The {\em order} of a truth value is defined as follows:
$order(T_\alpha) = \alpha$, $order(F_\alpha)  = \alpha$ and $order(0) = +\infty$.
\end{definition}

We can now define the meaning of all the types of our language as well as the corresponding
relations $\leq$ and $\sqsubseteq_\alpha$. This is performed in the following definitions:
\begin{definition}
We define the relation $\sqsubseteq_\alpha$ on the set $V$ for each $\alpha<\kappa$ as follows:
\begin{enumerate}
\item $x \sqsubseteq_\alpha x$ if $order(x) < \alpha$;
\item $F_\alpha \sqsubseteq_\alpha x$ and $x \sqsubseteq_\alpha T_\alpha$ if $order(x) \geq \alpha$;
\item $x \sqsubseteq_\alpha y$ if $order(x), order(y) > \alpha$.
\end{enumerate}
\end{definition}
Notice that $x =_\alpha y$ iff either $x = y$ or $order(x)>\alpha$
and $order(y) > \alpha$.


\begin{definition}\label{semantics-of-types}
Let $D$ be a nonempty set. Then:
\begin{itemize}
\item $\lsem \iota \rsem_D = D$, and $\leq_\iota$ is the trivial partial order
      such that $d \leq_\iota d$, for all $d \in D$;
\item $\lsem \iota^n \rightarrow \iota \rsem_D = D^n \rightarrow D$.
      A partial order in this case will not be needed;
\item $\lsem o \rsem_D = V$, and $\leq_o$ is the partial order of $V$;
\item $\lsem \iota \rightarrow \pi \rsem_D = D \rightarrow \lsem \pi \rsem_D$, and
      $\leq_{\iota \rightarrow \pi}$ is the partial order defined as follows:
      for all $f,g\in \lsem \iota \rightarrow \pi \rsem_D$, $f \leq_{\iota\rightarrow \pi} g$
      iff $f(d) \leq_\pi g(d)$ for all $d \in D$;
\item $\lsem \pi_1 \rightarrow \pi_2 \rsem_D = [\lsem \pi_1 \rsem_D \stackrel{m}{\rightarrow}  \lsem \pi_2 \rsem_D]$,
      and $\leq_{\pi_1 \rightarrow \pi_2}$ is the partial order
      defined as follows:
      for all $f,g\in \lsem \pi_1 \rightarrow \pi_2 \rsem_D$, $f \leq_{\pi_1 \rightarrow \pi_2} g$
      iff $f(d) \leq_{\pi_2} g(d)$ for all $d \in \lsem \pi_1 \rsem_D$.
\end{itemize}
\end{definition}

The subscripts in the above partial orders will often be omitted when they are obvious from context.

\begin{definition}
Let $D$ be a nonempty set and $\alpha < \kappa$. Then:
\begin{itemize}
\item The relation $\sqsubseteq_\alpha$ on $\lsem o \rsem_D$ is the relation $\sqsubseteq_\alpha$ on $V$.
\item The relation $\sqsubseteq_\alpha$ on $\lsem \rho \rightarrow \pi \rsem_D$ is defined as follows:
      $f \sqsubseteq_\alpha g$ iff
      $f(d) \sqsubseteq_\alpha g(d)$ for all $d \in \lsem \rho \rsem_D$. Moreover,
      $f \sqsubset_\alpha g$ iff $f \sqsubseteq_\alpha g$
      and $f(d) \sqsubset_\alpha g(d)$ for some $d \in \lsem \rho \rsem_D$.
\end{itemize}
\end{definition}


The following lemma expresses the fact that all the predicate types correspond to
semantic domains that are both complete lattices and basic models:
\begin{lemma}\label{pi-is-model}
Let $D$ be a nonempty set and $\pi$ be a predicate type. Then,
$(\lsem \pi \rsem_D, \leq_\pi)$ is a complete lattice and a basic model.
\end{lemma}

We now proceed to formally define the semantics of ${\cal H}$:
\begin{definition}
An intepretation $I$ of ${\cal H}$ consists of:
\begin{enumerate}
\item a nonempty set $D$ called the domain of $I$;
\item an assignment to each individual constant symbol $\mathsf{c}$, of an element
      $I(\mathsf{c}) \in D$;
\item an assignment to each predicate constant $\mathsf{p}:\pi$ of an element
      $I(\mathsf{p}) \in \lsem \pi \rsem_D$;
\item an assignment to each function symbol $\mathsf{f} : \iota^n \to \iota$
      of a function $I(\mathsf{f}) \in D^n\!\rightarrow D$.
\end{enumerate}
\end{definition}

\begin{definition}
Let $D$ be a nonempty set. A  state $s$ of ${\cal H}$ over $D$ is a function
that assigns to each argument variable $\mathsf{V}$ of type $\rho$ of ${\cal H}$,
of an element $s(\mathsf{V}) \in \lsem \rho \rsem_D$.
\end{definition}

\begin{definition}\label{definition-semantics-of-expressions}
Let $I$ be an interpretation of ${\cal H}$, let $D$ be the domain of
$I$, and let $s$ be a state over $D$. Then, the semantics of
expressions of ${\cal H}$ with respect to $I$ and $s$, is defined as
follows:
\begin{enumerate}
\item $\lsem \cfalse \rsem_s (I) = F_0$

\item $\lsem \ctrue \rsem_s (I) = T_0$

\item $\lsem \mathsf{c} \rsem_s (I) = I(\mathsf{c})$, for every
      individual constant $\mathsf{c}$

\item $\lsem \mathsf{p} \rsem_s (I) = I(\mathsf{p})$, for every
      predicate constant $\mathsf{p}$

\item $\lsem \mathsf{V} \rsem_s (I) = s(\mathsf{V})$, for every
      argument variable $\mathsf{V}$

\item $\lsem (\mathsf{f}\,\,\mathsf{E}_1\cdots \mathsf{E}_n) \rsem_s (I) =
      I(\mathsf{f})\,\,\lsem \mathsf{E}_1\rsem_s (I) \cdots \lsem \mathsf{E}_n\rsem_s (I)$,
      for every $n$-ary function symbol $\mathsf{f}$

\item $\lsem \mathsf{(}\mathsf{E}_1\mathsf{E}_2\mathsf{)} \rsem_s
      (I)= \lsem \mathsf{E}_1 \rsem_s (I)(\lsem \mathsf{E}_2\rsem_s(I))$

\item $\lsem \mathsf{(\lambda V.E)} \rsem_s (I) =\lambda d.\lsem
      \mathsf{E}\rsem_{s[\mathsf{V}/d]}(I)$, where $d$ ranges over
      $\lsem type(\mathsf{V})\rsem_D$

\item $\lsem (\mathsf{E}_1 \bigvee_{\pi} \mathsf{E}_2)\rsem_s (I) =
      \bigvee_{\pi}\{\lsem \mathsf{E}_1\rsem_s(I),\lsem \mathsf{E}_2\rsem_s(I)\}$, where
      $\bigvee_{\pi}$ is the least upper bound function on $\lsem \pi
      \rsem_D$

\item $\lsem (\mathsf{E}_1 \bigwedge_{\pi} \mathsf{E}_2)\rsem_s (I) =
      \bigwedge_{\pi}\{\lsem \mathsf{E}_1\rsem_s(I),\lsem \mathsf{E}_2\rsem_s(I)\}$, where
      $\bigwedge_{\pi}$ is the greatest lower bound function on $\lsem \pi
      \rsem_D$
\item $\lsem (\mnot \mathsf{E}) \rsem_s (I) =
          \begin{cases}
              T_{\alpha+1} & \mbox{if $\lsem \mathsf{E} \rsem_s(I) = F_\alpha$} \\
              F_{\alpha+1} & \mbox{if $\lsem \mathsf{E} \rsem_s(I) = T_\alpha$} \\
              0            & \mbox{if $\lsem \mathsf{E} \rsem_s(I) = 0$}
          \end{cases}$

\item $\lsem (\mathsf{E}_1 \,\mathsf{\approx}\, \mathsf{E}_2)\rsem_s (I) = \left\{\begin{array}{ll}
                                               T_0, & \mbox{if $\lsem \mathsf{E}_1 \rsem_s
                                                    (I) = \lsem \mathsf{E}_2 \rsem_s
                                                    (I)$}\\
                                               F_0, & \mbox{otherwise}
                                                   \end{array} \right. $

\item $\lsem (\exists \mathsf{V}\, \mathsf{E}) \rsem_s (I)= \bigvee_{d \in \lsem type(\mathsf{V}) \rsem_D} \lsem \mathsf{E} \rsem_{s[\mathsf{V}/d]}(I)$
\end{enumerate}
\end{definition}

For closed expressions $\mathsf{E}$ we will often write $\lsem
\mathsf{E} \rsem(I)$ instead of $\lsem \mathsf{E} \rsem_s(I)$
(since, in this case, the meaning of $\mathsf{E}$ is independent of
$s$).



\begin{lemma} \label{expressions-well-defined}
Let $\mathsf{E} : \rho$ be an expression and let $D$ be a nonempty set.
Moreover, let $s$ be a state over $D$ and let $I$ be an interpretation over $D$.
Then, $\lsem \mathsf{E} \rsem_s(I) \in \lsem \rho \rsem_D$.
\end{lemma}

\begin{definition}
Let $\mathsf{P}$ be a program and let $M$ be an interpretation over a nonempty set $D$.
Then $M$ will be called a {\em model} of $\mathsf{P}$ iff for all clauses
$\mathsf{p} \leftarrow_\pi \mathsf{E}$ of $\mathsf{P}$,
it holds $\lsem \mathsf{E} \rsem(M) \leq_\pi M(\mathsf{p})$, where $M(\mathsf{p}) \in \lsem \pi \rsem_D$.
\end{definition}

\section{Minimum Herbrand Model Semantics for ${\cal H}$}\label{section6}
In this section we demonstrate that every program of ${\cal H}$ has a unique
minimum Herbrand model which is the greatest lower bound of all the Herbrand models
of the program, and also the least fixed point of the immediate consequence operator
of the program. We start with the relevant definitions.
\begin{definition}
Let $\mathsf{P}$ be a program. The Herbrand universe $U_{\mathsf{P}}$ of $\mathsf{P}$ is the set of all terms
that can be formed out of the individual constants\footnote{As usual, if $\mathsf{P}$ has no constants, we assume
the existence of an arbitrary one.} and the function symbols of $\mathsf{P}$.
\end{definition}

\begin{definition}
A Herbrand interpretation $I$ of a program $\mathsf{P}$ is an interpretation such that:
\begin{enumerate}
\item the domain of $I$ is the Herbrand universe $U_{\mathsf{P}}$ of $\mathsf{P}$;
\item for every individual constant $\mathsf{c}$ of $\mathsf{P}$, $I(\mathsf{c})=\mathsf{c}$;
\item for every predicate constant $\mathsf{p}:\pi$ of $\mathsf{P}$, $I(\mathsf{p}) \in \lsem \pi \rsem_{U_{\mathsf{P}}}$;
\item for every $n$-ary function symbol $\mathsf{f}$ of $\mathsf{P}$ and for all $\mathsf{t}_1,\ldots,\mathsf{t}_n \in U_{\mathsf{P}}$, $I(\mathsf{f})\, \mathsf{t}_1\cdots \mathsf{t}_n = \mathsf{f}\,\mathsf{t}_1\cdots\mathsf{t}_n$.
\end{enumerate}
\end{definition}

A Herbrand state of a program $\mathsf{P}$ is a state whose underlying domain is $U_{\mathsf{P}}$.
We denote the set of Herbrand interpretations of a program $\mathsf{P}$ by ${\cal I}_\mathsf{P}$.

\begin{definition}
A Herbrand model of a program $\mathsf{P}$ is a Herbrand
interpretation that is a model of $\mathsf{P}$.
\end{definition}

\begin{definition}
Let $\mathsf{P}$ be a program. We define the following partial order on ${\cal I}_\mathsf{P}$:
for all $I, J \in  {\cal I}_\mathsf{P}$, $I \leq_{{\cal I}_\mathsf{P}} J$ iff for every
$\pi$ and for every predicate constant $\mathsf{p} : \pi$ of $\mathsf{P}$,
$I(\mathsf{p}) \leq_\pi J(\mathsf{p})$.
\end{definition}

\begin{definition}
Let $\mathsf{P}$ be a program. We define the following preorder on ${\cal I}_\mathsf{P}$
for all $\alpha < \kappa$:
for all $I, J \in  {\cal I}_\mathsf{P}$, $I \sqsubseteq_\alpha J$ iff for every
$\pi$ and for every predicate constant $\mathsf{p} : \pi$ of $\mathsf{P}$,
$I(\mathsf{p}) \sqsubseteq_\alpha J(\mathsf{p})$.
\end{definition}

The following two lemmas play a main role in establishing the two central theorems.
\begin{lemma}\label{interpretations-lattice}
Let $\mathsf{P}$ be a program.
Then, ${\cal I}_\mathsf{P}$ is a complete lattice and a basic model.
\end{lemma}

\begin{lemma}[$\alpha$-Monotonicity of Semantics]\label{monotonicity-of-semantics}
Let $\mathsf{P}$ be a program and let $\mathsf{E} : \pi$ be an expression.
Let $I, J$ be Herbrand interpretations and $s$ be a Herbrand
state of $\mathsf{P}$. For all $\alpha < \kappa$, if $I \sqsubseteq_\alpha J$ then
$\lsem \mathsf{E} \rsem_s(I) \sqsubseteq_\alpha \lsem \mathsf{E} \rsem_s(J)$.
\end{lemma}

%

Since by Lemma~\ref{interpretations-lattice} the set ${\cal I}_\mathsf{P}$ is a basic model
(and thus by Lemma~\ref{basic-model-lattice} is a complete lattice with respect to $\sqsubseteq$),
every ${\cal M}\subseteq {\cal I}_\mathsf{P}$ has a greatest lower bound $\bigsqcap{\cal M}$
with respect to $\sqsubseteq$. We have the following theorem which generalizes the familiar model
intersection theorem for definite first-order logic programs~\cite{lloyd}, the model intersection
theorem for normal first-order logic programs~\cite[Theorem~8.6]{tocl05} and the model intersection
theorem for definite higher-order logic programs~\cite[Theorem~6.8]{CHRW13}.
\begin{theorem}[Model Intersection Theorem]\label{model-intersection}
Let $\mathsf{P}$ be a program and ${\cal M}$ be a nonempty set of
Herbrand models of $\mathsf{P}$. Then, $\bigsqcap{\cal M}$ is also
a Herbrand model of $\mathsf{P}$.
\end{theorem}

\begin{definition}
Let $\mathsf{P}$ be a program.
The mapping $T_\mathsf{P} : {\cal I}_\mathsf{P} \rightarrow {\cal I}_\mathsf{P}$ is defined
for every $\mathsf{p} : \pi$ and for every $I \in {\cal I}_\mathsf{P}$ as
\(
T_{\mathsf{P}}(I)(\mathsf{p}) = \bigvee\{\lsem \mathsf{E} \rsem(I) : (\mathsf{p} \leftarrow_\pi \mathsf{E}) \in \mathsf{P}\}
\).
The mapping $T_\mathsf{P}$ will be called the {\em immediate consequence operator} for $\mathsf{P}$.
\end{definition}

The following two lemmas
are crucial in establishing the least fixed point theorem.
\begin{lemma}\label{denotation-lemma}
Let $\mathsf{P}$ be a program.
For every predicate constant $\mathsf{p}:\pi$ in $\mathsf{P}$ and $I \in {\cal I}_\mathsf{P}$, $T_P(I)(\mathsf{p}) \in \lsem \pi \rsem_{U_\mathsf{P}}$.
\end{lemma}

\begin{lemma}\label{tp-a-monotonic}
Let $\mathsf{P}$ be a program. Then, $T_\mathsf{P}$ is $\alpha$-monotonic for all $\alpha < \kappa$.
\end{lemma}



\begin{theorem}[Least Fixed Point Theorem]\label{least-model}
Let $\mathsf{P}$ be a program and let ${\cal M}$ be the set of all its Herbrand models.
Then, $T_\mathsf{P}$ has a least fixed point $M_\mathsf{P}$. Moreover, $M_\mathsf{P} = \bigsqcap{\cal{M}}$.
\end{theorem}

The construction of the least fixed point in the above theorem is similar to the one
given for (potentially infinite) propositional programs in~\cite[Section~6]{tocl05}.
Due to space limitations, we provide a short outline of this procedure. In order to
calculate the least fixed point, we start with an interpretation, say $I_0$, which for every
predicate constant $\mathsf{p}$ of type $\rho_1\rightarrow \cdots \rho_n \rightarrow o$,
and for all $d_1\in\lsem \rho_1\rsem_{U_{\mathsf{P}}},\ldots,d_n\in\lsem \rho_n\rsem_{U_{\mathsf{P}}}$,
$I_0(\mathsf{p})\,d_1\cdots d_n = F_0$. We start iterating $T_{\mathsf{P}}$ on this interpretation
until we get to a point where the additional iterations do not affect the $F_0$ and $T_0$ values.
At this point, we reset all the remaining values (regarding predicate constants and arguments
that have not stabilized) to $F_1$, getting an interpretation $I_1$. We start iterating $T_{\mathsf{P}}$ on $I_1$,
until we get to a point where the additional iterations do not affect the $F_1$ and $T_1$ values.
We repeat this process for higher ordinals. In particular, when we get to a limit ordinal, say $\alpha$,
we reset all the values that have not stabilized to a truth value of order less than $\alpha$,
to $F_\alpha$. The whole process is repeated for $\kappa$ times. If the value of certain predicate
constants applied to certain arguments has not stabilized after the $\kappa$ iterations, we assign
to them the intermediate value 0. The resulting interpretation is the least fixed point $M_{\mathsf{P}}$.

\section{Resolving a Semantic Paradox of Higher-Order Logic Programming}
One deficiency of extensional higher-order logic programming is the
inability to define rules (or facts) that have predicate constants in their heads.
The reason of this restriction is a semantic one and will be explained shortly. However,
not all programs that use predicate constants in the heads of clauses are problematic.
For example, the program
\[
\begin{array}{l}
\mbox{\tt computer\_scientist(john).}\\
\mbox{\tt good\_profession(computer\_scientist).}
\end{array}
\]
has a clear declarative reading: the denotation of the {\tt computer\_scientist} predicate is
the relation $\{{\tt john}\}$, while the denotation of {\tt good\_profession} is
the relation $\{\{{\tt john}\}\}$.

In~\cite{Wadge91}, W. W. Wadge argued that allowing rules to have predicate constants
in their heads, creates tricky semantic problems to.
Wadge gave a simple example (duplicated below) that revealed these problems; the example
has since been used in other studies of higher-order logic programming (such as for example
in~\cite{Bezem2}). We present the example in almost identical phrasing as it initially appeared.
\begin{example}\label{wadge-example}
Consider the program:
\[
\begin{array}{l}
\mbox{\tt p(a).}\\
\mbox{\tt q(a).} \\
\mbox{\tt phi(p).}\\
\mbox{\tt q(b):-phi(q).}
\end{array}
\]
One candidate for minimum Herbrand model is the one in which {\tt p} and {\tt q} are true only
of {\tt a}, and {\tt phi} is true only of {\tt p}. However, this means that {\tt p} and {\tt q}
have the same extension, and so themselves are equal. But since {\tt p} and {\tt q} are equal,
and {\tt phi} holds for {\tt p}, it must also hold for {\tt q}. The fourth rule forces us to
add {\tt q(b)}, so that the model becomes $\{\mbox{\tt p(a)},\mbox{\tt phi(p)},\mbox{\tt q(a)},\mbox{\tt q(b)}\}$
(in ad hoc notation). But this is problematic because {\tt p} and {\tt q} are no longer equal
and {\tt q(b)} has lost its justification.
\end{example}

Problems such as the above led Wadge to disallow such clauses from the syntax of the
language proposed in~\cite{Wadge91}. Similarly, the higher-order language introduced
in~\cite{CHRW13} also disallows this kind of clauses.

However, under the semantics presented in this paper, we can now assign a proper meaning to programs
such as the above. Actually, higher order facts such as {\tt phi(p).} above, can be seen
as syntactic sugar in our fragment. A fact of this form simply states that {\tt phi} is
true of a relation if this relation {\em is equal} to {\tt p}. This can simply be written
as:
\[
\begin{array}{l}
\mbox{\tt phi(P):-equal(P,p).}
\end{array}
\]
where {\tt equal} is a higher-order equality relation that can easily be axiomatized in ${\cal H}$ using
the {\tt subset} predicate (see Example~\ref{subset-predicate-example}):
\[
\mathtt{equal} \leftarrow
    \lambda\mathtt{P}\per \lambda\mathtt{Q}.\lpa\mathtt{subset}\,\,\mathtt{P}\,\,\mathtt{Q}\rpa \wedge
                                            \lpa\mathtt{subset}\,\,\mathtt{Q}\,\,\mathtt{P}\rpa.
\]
%

%

One can compute the minimum model of the resulting program using the techniques presented
in this paper. The paradox of Example~\ref{wadge-example} is no longer valid since in the minimum
infinite-valued model the atom {\tt q(b)} has value 0. Intuitively, this means that it is not
possible to decide whether {\tt q(b)} should be true or false.

The above discussion leads to an easy way of handling rules with predicate constants
in their heads. The predicate constants are replaced with predicate variables and higher-order equality atoms
are added in the bodies of clauses. Then, appropriate clauses defining the {\tt equal} predicates
for all necessary types, are added to the program. The infinite valued semantics of the resulting program
is taken as the meaning of the initial program.

\section{Future Work}
We have presented the first, to our knowledge, formal semantics for negation in extensional
higher-order logic programming. The results we have obtained generalize the semantics of classical
logic programming to the higher order setting. We believe that the most interesting direction for
future work is the investigation of implementation techniques for (fragments of) ${\cal H}$, based
on the semantics introduced in this paper. One possible option would be to examine the implementation
of a higher order extension of Datalog with negation. We are currently examining these possibilities.

\bibliographystyle{acmtrans}
\bibliography{iclp2014}

\appendix

\def\pimodel{\ref{pi-is-model}}
\section{Proof of Lemma~\pimodel{}}\label{appendix-A}

We will make use of certain facts established in \cite{tocl14}.

Suppose that $L$ is a basic model. For each $x \in L$ and $\alpha < \kappa$,
we define $x|_\alpha = \bigsqcup_\alpha \{x\}$. It was shown in \cite{tocl14}
that $x =_\alpha x|_\alpha$ and $x|_\alpha =_\alpha x|_\beta$,
$x|_\alpha \leq x|_\beta$ for all
$\alpha < \beta < \kappa$.
Moreover, $x = \bigvee_{\alpha < \kappa} x|_\alpha$.
Also, for all $x,y \in L$ and $\alpha < \kappa$, it holds
 $x =_\alpha y$ iff $x|_\alpha =_\alpha y|_\alpha$ iff
$x|_\alpha = y|_\alpha$,  and $x \sqsubseteq_\alpha y$ iff
$x|_\alpha \sqsubseteq_\alpha y|_\alpha$. And if
$x \sqsubseteq_\alpha y$, then $x|_\alpha \leq y|_\alpha$.
It is also not difficult to prove that
for all  $x\in L$ and $\alpha,\beta<\kappa$,
$(x|_\alpha)|_\beta = x|_{\min\{\alpha,\beta\}}$.
More generally, whenever $X \subseteq (z]_\alpha$ and $\beta \leq \alpha <\kappa$,
it holds $(\bigsqcup_\alpha X)|_\beta = \bigsqcup_\beta X$.
And if $\alpha <\beta$, then $(\bigsqcup_\alpha X)|_\beta = \bigsqcup_\alpha X$.
Finally, we will make use of the following two results from~\cite{tocl14}:
\begin{proposition}\label{lub-of-a-monotonic}
Let $A, B$ be basic models and let $\alpha < \kappa$.
If $f_j : A \rightarrow B$ is an $\alpha$-monotonic function for each $j \in J$,
then so is $f = \bigvee_{j \in J}f_j$ defined by $f(x) = \bigvee_{j \in J}{f_j(x)}$.
\end{proposition}
\begin{lemma}\label{function-model}
Let $Z$ be an arbitrary set and $L$ be a basic model. Then, $Z \rightarrow L$ is a basic model
with the pointwise definition of the order of relations $\leq$ and $\sqsubseteq_\alpha$ for all $\alpha < \kappa$.
\end{lemma}

Suppose that $A,B$ are basic models. By Lemma~\ref{function-model}
the set $A \rightarrow B$ is also a model, where the relations
$\leq$ and $\sqsubseteq_\alpha$, $\alpha < \kappa$, are
defined in a pointwise way (see~\cite[Subsection 5.3]{tocl14} for details).
It follows that for any set $F$
of functions $A \to B$, $\bigvee F $ can be computed pointwisely.
Also, when $F\subseteq (f]_\alpha$ for some $f: A \to B$,
$\bigsqcup_\alpha F$
for $\alpha < \kappa$ can be computed
pointwisely.

We want to show that whenever $f : A \to B$, $\beta < \kappa$
and $F \subseteq (f]_\beta$ is a set of functions such that
$F\subseteq [A  \stackrel{m}{\rightarrow} B]$,
then $\bigsqcup_\beta F \in [A  \stackrel{m}{\rightarrow} B]$.
We will make use of a lemma.

\begin{lemma}
\label{lem-appA}
Let $L$ be a basic model.
For all $x,y\in L $ and
$\alpha ,\beta < \kappa$ with $\alpha \neq\beta$,
$x|_\beta \sqsubseteq_\alpha y|_\beta$ iff
either $\beta < \alpha$ and $x|_\beta = y|_\beta$ (or
equivalently, $x =_\beta y$),
or $\beta > \alpha$ and $x|_\alpha \sqsubseteq_\alpha y|_\alpha$.
\end{lemma}

\begin{proof}
Let $x|_\beta \sqsubseteq_\alpha y|_\beta$.
If $\beta < \alpha$ then $x|_\beta = (x|_\beta)|_\beta
= (y|_\beta)|_\beta = y|_\beta$. If $\beta > \alpha$
then $x|_\alpha = (x|_\beta)|_\alpha \sqsubseteq_\alpha
(y|_\beta)|_\alpha =y|_\alpha$.

Suppose now that $\beta < \alpha$ and $x|_\beta = y|_\beta$.
Then $(x|_\beta)|_\alpha = x|_\beta = y|_\beta
= (y|_\beta)|_\alpha$ and thus $x|_\beta =_\alpha y|_\beta$.
Finally, let $\beta > \alpha$ and $x|_\alpha \sqsubseteq_\alpha y|_\alpha$.
Then $(x|_\beta)|_\alpha =x|_\alpha \sqsubseteq_\alpha  y|_\alpha
=  (y|_\beta)|_\alpha$ and thus $x|_\beta \sqsubseteq_\alpha y|_\beta$.
\end{proof}

\begin{remark}
Under the above assumptions, if $\beta < \alpha$,
then $x|_\beta \sqsubseteq_\alpha y|_\beta$ iff $x|_\beta =_\alpha y|_\beta$ iff
$x|_\beta = y|_\beta$.
\end{remark}

\begin{corollary}
\label{cor-appA}
For all $X,Y \subseteq L$ and $\alpha \neq \beta$,
$\bigsqcup_\beta X \sqsubseteq_\alpha \bigsqcup_\beta Y$ iff
$\beta < \alpha$ and  $\bigsqcup_\beta X = \bigsqcup_\beta Y$,
or $\beta> \alpha$ and $\bigsqcup_\alpha X \sqsubseteq_\alpha \bigsqcup_\alpha Y$.
\end{corollary}

\begin{proof}
Let $x = \bigsqcup_\beta X$ and $y = \bigsqcup_\beta Y$. Then $x = \bigsqcup_\beta X =
\bigsqcup_\beta \{  \bigsqcup_\beta X \} =x|_\beta$ and $y = y|_\beta$.
Let $\beta < \alpha$. Then $x\sqsubseteq_\alpha y$ iff $x = y$.
Let $\beta >\alpha$. Then $x\sqsubseteq_\alpha y$ iff
$x|_\alpha \sqsubseteq_\alpha y|_\alpha$.
But $x|_\alpha = \bigsqcup_\alpha\{\bigsqcup_\beta X\}
= \bigsqcup_\alpha X$ and similarly for $Y$.
\end{proof}

\begin{lemma}\label{lub-beta-lemma}
Let $A$ and $B$ be basic models.
Suppose that $f: A \to B$ and $F \subseteq (f]_\beta
$ (where $\beta < \kappa$) is a set of functions in $[A  \stackrel{m}{\rightarrow} B]$.
Then $\bigsqcup_\beta F$ is also $\alpha$-monotonic for all $\alpha < \kappa$.
\end{lemma}

\begin{proof}
Suppose that $\alpha,\beta < \kappa$ and  $x \sqsubseteq_\alpha y$ in $A$.
Then $(\bigsqcup_\beta F)(x) =
\bigsqcup_\beta \{f(x) : f \in F\}$ and
 $(\bigsqcup_\beta F)(y) =
\bigsqcup_\beta \{f(y) : f \in F\}$.
 We have that
$f(x)\sqsubseteq_\alpha f(y)$ for all $f \in F$.
Thus, if  $\alpha = \beta$, then clearly
$(\bigsqcup_\beta F)(x) \sqsubseteq_\alpha (\bigsqcup_\beta F)(y)$.

Suppose that $\beta < \alpha$. Then
$\bigsqcup_\beta \{f(x) : f \in F\} = \bigsqcup_\beta \{f(y) : f \in F\}$
since $f(x) =_\beta f(y)$ for all $f \in F$. Thus, by Corollary~\ref{cor-appA},
$(\bigsqcup_\beta F)(x) \sqsubseteq_\alpha (\bigsqcup_\beta F)(y)$.

Suppose that $\beta > \alpha$. Then $(\bigsqcup_\beta F)(x) \sqsubseteq_\alpha
(\bigsqcup_\beta F)(y)$ follows by Corollary~\ref{cor-appA} from
$\bigsqcup_\alpha \{f(x) : f \in F\}
\sqsubseteq_\alpha \bigsqcup_\alpha \{f(y) : f \in F\}$.
\end{proof}

We equip $[A  \stackrel{m}{\rightarrow} B]$ with the order relations $\leq$ and $\sqsubseteq_\alpha$
inherited from $A \rightarrow B$. We have the following lemma:

\begin{lemma}\label{lem-exp2}
If $A$ and $B$ are basic models, then so is $[A  \stackrel{m}{\rightarrow} B]$
with the pointwise definition of the order of relations $\leq$ and $\sqsubseteq_\alpha$ for all $\alpha < \kappa$.
\end{lemma}

\begin{proof}
It is proved in \cite{tocl14} that the set of functions
$A \to B$  is a basic model with the pointwise definition of the
relations $\leq$ and $\sqsubseteq_\alpha$,
so that for all $f,g: A \to B$ and $\alpha < \kappa$,
$f \leq g$ iff $ f(x) \leq g(x)$ for all $x \in A$
and $f\sqsubseteq_\alpha g$ iff $f(x) \sqsubseteq_\alpha g(x)$
for all $x \in A$. It follows that for any $F \subseteq B^A$
and $\alpha <\kappa$, $\bigvee F$ and $\bigsqcup_\alpha F$
can also be computed pointwise: $(\bigvee F)(x) =
\bigvee \{f(x) : x \in A\}$ and $(\bigsqcup_\alpha F)(x)
= \bigsqcup_\alpha\{f(x) : f \in F\}$.
By Proposition~\ref{lub-of-a-monotonic} and Lemma~\ref{lub-beta-lemma}, for all $F\subseteq B^A$, if $F$
is a set of functions $\alpha$-monotonic for all $\alpha$,
then $\bigvee F$ and $\bigsqcup_\beta F$ are also $\alpha$-monotonic
for all $\alpha$. Since the relations $\leq$ and $\sqsubseteq_\alpha$, $\alpha <
\kappa$  on
$[A  \stackrel{m}{\rightarrow} B]$ are the restrictions of the corresponding relations on $B^A$,
in view of Proposition~\ref{lub-of-a-monotonic} and Lemma~\ref{lub-beta-lemma},
$[A  \stackrel{m}{\rightarrow} B]$ also satisfies the
axioms in Definition~\ref{def-basic model},
so that $[A  \stackrel{m}{\rightarrow} B]$ is a basic model.
\end{proof}

The following lemma is shown in~\cite[Subsection 5.2]{tocl14} and will be used
in the proof of the basis case of the next lemma:
\begin{lemma}
$(V, \leq)$ is a complete lattice and a basic model.
\end{lemma}

\begingroup
\def\thelemma{\ref{pi-is-model}}
\begin{lemma}
Let $D$ be a nonempty set and $\pi$ be a predicate type. Then,
$(\lsem \pi \rsem_D, \leq_\pi)$ is a complete lattice and a basic model.
\end{lemma}
\addtocounter{lemma}{-1}
\endgroup
\begin{proof}
Let $\pi$ be a predicate type. We prove that $\lsem \pi \rsem_D$
is a basic model by induction on the structure of $\pi$.
When $\pi = o $, $\lsem \pi \rsem_D = V$,
a basic model. Suppose that $\pi$ is of the sort $\iota \to \pi'$. Then
$\lsem \pi \rsem_D = D \rightarrow \lsem \pi'\rsem_D$, which is a basic model,
since $\lsem \pi' \rsem_D$ is a model by the induction hypothesis.
 Finally,
let $\pi$ be of the sort $\pi_1 \to \pi_2$. By the induction
hypothesis, $\lsem \pi_i \rsem_D$ is a model for $i = 1,2$.
Thus, by Lemma~\ref{lem-exp2}, $\lsem \pi \rsem_D = [\lsem \pi_1\rsem_D \stackrel{m}{\rightarrow}
\lsem \pi_2 \rsem_D ]$ is also a basic model.
\end{proof}

\begin{remark}
Let $\C$ denote the category of all basic models
and $\alpha$-monotonic functions. The above results
show that $\C$ is cartesian closed, since
for all basic models $A,B$,
the evaluation function $\eval: (A \times B) \times A \to B$
is $\alpha$-monotonic (in both arguments) for all $\alpha < \kappa$.

Indeed, suppose that $f,g\in [A \stackrel{m}{\to} B]$ and $x,y \in A$ with
$f\sqsubseteq_\alpha g$ and $x\sqsubseteq_\alpha y$.
Then $\eval(f,x) = f(x) \sqsubseteq_\alpha g(x) = \eval(g,x)$
by the pointwise definition of $f \sqsubseteq_\alpha g$.
Also, $\eval(f,x) = f(x) \sqsubseteq_\alpha f(y) = \eval(f,y)$
since $f$ is $\alpha$-monotonic.

Since $\C$ is cartesian closed, for all $f\in [B \times A \stackrel{m}{\to} C]$
there is a unique $\Lambda f\in\![B \stackrel{m}{\to}\![A \stackrel{m}{\to}\!C]]$ in with
$f(y,x) = \eval(\Lambda f (y), x)$ for all $x \in A$ and $y \in B$.
\end{remark}

\def\lemmaexp{\ref{expressions-well-defined}}
\def\monotonicitylemma{\ref{monotonicity-of-semantics}}
\def\interlatticelemma{\ref{interpretations-lattice}}
\section{Proofs of Lemmas~\lemmaexp{}, \interlatticelemma{} and \monotonicitylemma{}}

\begingroup
\def\thelemma{\ref{expressions-well-defined}}
\begin{lemma}
Let $\mathsf{E} : \rho$ be an expression and let $D$ be a nonempty set.
Moreover, let $s$ be a state over $D$ and let $I$ be an interpretation over $D$.
Then, $\lsem \mathsf{E} \rsem_s(I) \in \lsem \rho \rsem_D$.
\end{lemma}
\addtocounter{lemma}{-1}
\endgroup
\begin{proof}
If $\rho = \iota$ then the claim is clear. Let $\mathsf{E}$ be of
a predicate type $\pi$.
We prove simultaneously the following auxiliary statement.
Let $\alpha < \kappa$, $\mathsf{V} : \pi$, $x, y \in \lsem \pi' \rsem_D$.
If $x\sqsubseteq_\alpha y$ then $\lsem \mathsf{E} \rsem_{s[\mathsf{V}/x]}(I) \sqsubseteq_\alpha \lsem \mathsf{E} \rsem_{s[\mathsf{V}/y]}(I)$.
The proof is by structural induction on $\mathsf{E}$.
We will cover only the nontrivial cases.

\vspace{0.1cm}
\noindent{\em Case $(\mathsf{E}_1\ \mathsf{E}_2)$:}
The main statement follows directly by the induction hypothesis of $\mathsf{E}_1$
and $\mathsf{E}_2$. There are two cases. Suppose that $E_1: \pi_1 \to \pi$ and
$E_2: \pi_1$. Then $\lsem \mathsf{E}_1 \rsem_s(I) \in \lsem \pi_1\to\pi\rsem_D
= [\lsem \pi_1 \rsem_D \stackrel{m}{\to} \lsem \pi \rsem_D ]$
and $\lsem \mathsf{E}_2 \rsem_s(I) \in \lsem \pi_1 \rsem_D$ by the induction
hypothesis. Thus,
$\lsem \mathsf{E}_1 \rsem_s(I)\ (\lsem \mathsf{E} \rsem_s(I))  \in \lsem \pi \rsem_D$.
Suppose now that $E_1: \iota \to \pi$ and $E_2: \iota$. Then
$ \lsem \mathsf{E}_1 \rsem_s(I)\in \lsem \iota \to \pi\rsem_D = D \to \lsem \pi \rsem_D$ by the induction
hypothesis and $\lsem \mathsf{E}_2 \rsem_s(I) \in \lsem \iota \rsem_D = D$.
It follows again that
$ \lsem \mathsf{E}_1 \rsem_s(I)\ (\lsem \mathsf{E} \rsem_s(I))  \in \lsem \pi \rsem_D$.

\vspace{0.1cm}
\noindent{\em Auxiliary statement:}
Let $x,y \in \lsem \pi' \rsem_D$ and assume $x \sqsubseteq_\alpha y$.
We have by definition
$\lsem (\mathsf{E}_1\ \mathsf{E}_2) \rsem_{s[\mathsf{V}/x]}(I) =
\lsem \mathsf{E}_1 \rsem_{s[\mathsf{V}/x]}(I)\
(\lsem \mathsf{E}_2 \rsem_{s[\mathsf{V}/x]}(I))$, and similarly for
$\lsem (\mathsf{E}_1\ \mathsf{E}_2) \rsem_{s[\mathsf{V}/y]}(I)$.
We have $E_1: \pi_1 \to \pi$ and $E_2: \pi_1$ or
$E_1: \iota \to \pi$ and $E_2: \iota$. In the first case,
by induction hypothesis $\lsem \mathsf{E}_1 \rsem_{s[\mathsf{V}/x]}(I) \in \lsem \pi_1 \rightarrow \pi \rsem_D$,
and thus is $\alpha$-monotonic. Also,
$\lsem \mathsf{E}_1 \rsem_{s[\mathsf{V}/x]}(I) \sqsubseteq_\alpha \lsem \mathsf{E}_1 \rsem_{s[\mathsf{V}/y]}(I)$
and
$\lsem \mathsf{E}_2 \rsem_{s[\mathsf{V}/x]}(I) \sqsubseteq_\alpha \lsem \mathsf{E}_2 \rsem_{s[\mathsf{V}/y]}(I)$ by the induction hypothesis.
It follows that
$$\lsem \mathsf{E}_1 \rsem_{s[\mathsf{V}/x]}(I)\
(\lsem \mathsf{E}_2 \rsem_{s[\mathsf{V}/x]}(I))
\sqsubseteq_\alpha
\lsem \mathsf{E}_1 \rsem_{s[\mathsf{V}/x]}(I)\
(\lsem \mathsf{E}_2 \rsem_{s[\mathsf{V}/y]}(I))
\sqsubseteq_\alpha
\lsem \mathsf{E}_1 \rsem_{s[\mathsf{V}/y]}(I)\
(\lsem \mathsf{E}_2 \rsem_{s[\mathsf{V}/y]}(I)).$$
The second case is similar. We have
$\lsem \mathsf{E}_1 \rsem_{s[\mathsf{V}/x]}(I) \sqsubseteq_\alpha \lsem \mathsf{E}_1 \rsem_{s[\mathsf{V}/y]}(I)$ by the induction hypothesis, moreover,
$\lsem \mathsf{E}_2 \rsem_{s[\mathsf{V}/x]}(I) = \lsem \mathsf{E}_2 \rsem_{s[\mathsf{V}/y]}(I)$.
Therefore,
$\lsem \mathsf{E}_1 \rsem_{s[\mathsf{V}/x]}(I)\
(\lsem \mathsf{E}_2 \rsem_{s[\mathsf{V}/x]}(I))
\sqsubseteq_\alpha
\lsem \mathsf{E}_1 \rsem_{s[\mathsf{V}/y]}(I)\
(\lsem \mathsf{E}_2 \rsem_{s[\mathsf{V}/y]}(I))$.

\vspace{0.1cm}
\noindent{\em Case $(\lambda\mathsf{V}.\mathsf{E})$}: Assume $V : \rho_1$
and $\mathsf{E}: \pi_2$. We will show that
$\lsem \lambda\mathsf{V}.\mathsf{E} \rsem_s(I) \in \lsem \rho_1 \rightarrow \pi_2 \rsem_D$.
If $\rho_1 = \iota$ then the result follows easily from the induction hypothesis of the first statement.
Assume $\rho_1 = \pi_1$. We show that
$\lsem \lambda\mathsf{V}.\mathsf{E} \rsem_s(I) \in \lsem \pi_1 \rightarrow \pi_2 \rsem_D$, that is,
$\lambda d. \lsem \mathsf{E} \rsem_{s[\mathsf{V}/d]}(I)$ is $\alpha$-monotonic for all $\alpha < \kappa$.
That follows directly by the induction hypothesis of the auxiliary statement.

\vspace{0.1cm}
\noindent{\em Auxiliary statement:}
It suffices to show that $\lsem (\lambda\mathsf{U}.\mathsf{E}) \rsem_{s[\mathsf{V}/x]}(I) \sqsubseteq_\alpha
\lsem (\lambda\mathsf{U}.\mathsf{E}) \rsem_{s[\mathsf{V}/y]}(I)$ and equivalently for every $d$,
$\lsem \mathsf{E} \rsem_{s[\mathsf{V}/x][\mathsf{U}/d]}(I) \sqsubseteq_\alpha
 \lsem \mathsf{E} \rsem_{s[\mathsf{V}/y][\mathsf{U}/d]}(I)$ which follows from induction hypothesis.
\end{proof}

\begingroup
\def\thelemma{\ref{interpretations-lattice}}
\begin{lemma}
Let $\mathsf{P}$ be a program.
Then, ${\cal I}_\mathsf{P}$ is a complete lattice and a basic model.
\end{lemma}
\addtocounter{lemma}{-1}
\endgroup
\begin{proof}
From Lemma~\ref{pi-is-model} we have that for all predicate types $\pi$, $\lsem \pi \rsem_{U_\mathsf{P}}$ is a complete lattice and a basic model. It follows, by Lemma~\ref{function-model}, that for all predicate types $\pi$, ${\cal P}_\pi \to \lsem \pi \rsem_{U_\mathsf{P}}$ is also a complete lattice
and a model, where ${\cal P}_\pi$ is the set of predicate constants of type $\pi$. Then, ${\cal I}_\mathsf{P}$ is $\prod_{\pi}{{\cal P}_\pi \to \lsem \pi \rsem_{U_\mathsf{P}}}$ which is also a basic model (proved in \cite{tocl14}).
\end{proof}

\begingroup
\def\thelemma{\ref{monotonicity-of-semantics}}
\begin{lemma}[$\alpha$-Monotonicity of Semantics]
Let $\mathsf{P}$ be a program and let $\mathsf{E} : \pi$ be an expression.
Let $I, J$ be Herbrand interpretations and $s$ be a Herbrand
state of $\mathsf{P}$. For all $\alpha < \kappa$, if $I \sqsubseteq_\alpha J$ then
$\lsem \mathsf{E} \rsem_s(I) \sqsubseteq_\alpha \lsem \mathsf{E} \rsem_s(J)$.
\end{lemma}
\addtocounter{lemma}{-1}
\endgroup

\begin{proof}
The proof is by structural induction on $\mathsf{E}$.

\vspace{0.1cm}
\noindent{\em Induction Base:}
The cases $\mathsf{V}, \mathsf{false}, \mathsf{true}$ are straightforward since
their meanings do not depend on $I$. Let $I \sqsubseteq_\alpha J$.
If $\mathsf{E}$ is a predicate constant $\mathsf{p}$ then we have $I(\mathsf{p}) \sqsubseteq_\alpha J(\mathsf{p})$.

\vspace{0.1cm}
\noindent{\em Induction Step:}
Assume that the statement holds for expressions $\mathsf{E}_1$ and $\mathsf{E}_2$ and let
$I \sqsubseteq_\alpha J$.

\vspace{0.1cm}
\noindent{Case $(\mathsf{E}_1\,\,\mathsf{E}_2)$:}
It holds $\lsem (\mathsf{E}_1\,\,\mathsf{E}_2) \rsem_s(I) = \lsem \mathsf{E}_1 \rsem_s(I)
(\lsem \mathsf{E}_2 \rsem_s(I))$. By induction hypothesis we have $\lsem \mathsf{E}_1 \rsem_s(I) \sqsubseteq_\alpha \lsem \mathsf{E}_1 \rsem_s(J)$ and therefore
$\lsem \mathsf{E}_1 \rsem_s(I)(\lsem \mathsf{E}_2 \rsem_s(I)) \sqsubseteq_\alpha
\lsem \mathsf{E}_1 \rsem_s(J)(\lsem \mathsf{E}_2 \rsem_s(I))$. We perform a case analysis on the type of
$\mathsf{E}_2$. If $\mathsf{E}_2$ is of type $\iota$ and since $I,J$ are Herbrand interpretations, it is clear that
$\lsem \mathsf{E}_2 \rsem_s(I) = \lsem \mathsf{E}_2 \rsem_s(J)$ and therefore
$\lsem \mathsf{E}_1 \rsem_s(I)(\lsem \mathsf{E}_2 \rsem_s(I)) \sqsubseteq_\alpha
\lsem \mathsf{E}_1 \rsem_s(J)(\lsem \mathsf{E}_2 \rsem_s(J))$. By definition of application we get
$\lsem (\mathsf{E}_1\,\,\mathsf{E}_2) \rsem_s(I) \sqsubseteq_\alpha \lsem (\mathsf{E}_1\,\,\mathsf{E}_2) \rsem_s(J)$.
If $\mathsf{E}_2$ is of type $\pi$ then by induction hypothesis we have
$\lsem \mathsf{E}_2 \rsem_s(I) \sqsubseteq_\alpha \lsem \mathsf{E}_2 \rsem_s(J)$ and since
$\lsem \mathsf{E}_1 \rsem_s(J)$ is $\alpha$-monotonic we get that
$\lsem \mathsf{E}_1 \rsem_s(J)(\lsem \mathsf{E}_2 \rsem_s(I))
\sqsubseteq_\alpha \lsem \mathsf{E}_1 \rsem_s(J)(\lsem \mathsf{E}_2 \rsem_s(J))$.
By transitivity of $\sqsubseteq_\alpha$ and by the definition of application we conclude that
$\lsem (\mathsf{E}_1\,\,\mathsf{E}_2) \rsem_s(I) \sqsubseteq_\alpha \lsem (\mathsf{E}_1\,\,\mathsf{E}_2) \rsem_s(J)$.

\vspace{0.1cm}
\noindent{Case $(\lambda\mathsf{V}.\mathsf{E}_1)$:}
It holds by definition that $\lsem (\lambda\mathsf{V}.\mathsf{E}_1) \rsem_s(I) = \lambda d.\lsem \mathsf{E}_1 \rsem_{s[\mathsf{V}/d]}(I)$.
It suffices to show that $\lambda d.\lsem \mathsf{E}_1 \rsem_{s[\mathsf{V}/d]}(I) \sqsubseteq_\alpha \lambda d. \lsem \mathsf{E}_1 \rsem_{s[\mathsf{V}/d]}(J)$
and equivalently that for every $d$, $\lsem \mathsf{E}_1 \rsem_{s[\mathsf{V}/d]}(I) \sqsubseteq_\alpha \lsem \mathsf{E}_1 \rsem_{s[\mathsf{V}/d]}(J)$
which holds by induction hypothesis.

\vspace{0.1cm}
\noindent{Case $(\mathsf{E}_1 \bigvee_\pi \mathsf{E}_2)$:}
It holds $\lsem (\mathsf{E}_1 \bigvee_\pi \mathsf{E}_2) \rsem_s(I) = \bigvee\{\lsem \mathsf{E}_1 \rsem_s(I), \lsem \mathsf{E}_2 \rsem_s(I) \}$. It suffices
to show that $\bigvee\{\lsem \mathsf{E}_1 \rsem_s(I), \lsem \mathsf{E}_2 \rsem_s(I) \} \sqsubseteq_\alpha \bigvee\{\lsem \mathsf{E}_1 \rsem_s(J), \lsem \mathsf{E}_2 \rsem_s(J) \}$ which holds by induction hypothesis and Axiom~\ref{axiom6}.

\vspace{0.1cm}
\noindent{Case $(\mathsf{E}_1 \bigwedge_\pi \mathsf{E}_2)$:}
It holds $\lsem (\mathsf{E}_1 \bigwedge_\pi \mathsf{E}_2) \rsem_s(I) = \bigwedge\{\lsem \mathsf{E}_1 \rsem_s(I), \lsem \mathsf{E}_2 \rsem_s(I) \}$.
Let $\pi = \rho_1 \to \cdots \to \rho_n \to o$, it suffices to show for all $d_i \in \lsem \rho_i \rsem_{U_{\mathsf{P}}}$,
$\bigwedge\{\lsem \mathsf{E}_1 \rsem_s(I)\ d_1\cdots d_n, \lsem \mathsf{E}_2 \rsem_s(I)\ d_ 1\cdots d_n \}\sqsubseteq_\alpha
\bigwedge\{\lsem \mathsf{E}_1 \rsem_s(J)\ d_1\cdots d_n, \lsem \mathsf{E}_2 \rsem_s(J)\ d_1\cdots d_n \}$.
We define $x_i = \lsem \mathsf{E}_i \rsem_s(I)\ d_1\cdots d_n$ and $y_i = \lsem \mathsf{E}_i \rsem_s(J)\ d_1\cdots d_n$ for $i \in \{1,2\}$.
We perform a case analysis on $v = \bigwedge\{x_1,x_2\}$. If $v < F_\alpha$ or $v > T_\alpha$ then $\bigwedge\{x_1,x_2\} = \bigwedge\{y_1,y_2\}$
and thus $\bigwedge\{x_1,x_2\} \sqsubseteq_\alpha \bigwedge\{y_1,y_2\}$. If $v = F_\alpha$ then $F_\alpha \leq \bigwedge\{y_1,y_2\} \leq T_\alpha$
and therefore $\bigwedge\{x_1,x_2\} \sqsubseteq_\alpha \bigwedge\{y_1,y_2\}$. If $v = T_\alpha$ then $\bigwedge\{y_1,y_2\} = T_\alpha$ and thus
$\bigwedge\{x_1,x_2\} \sqsubseteq_\alpha \bigwedge\{y_1,y_2\}$. If $F_\alpha < v < T_\alpha$ then $F_\alpha < \bigwedge\{y_1,y_2\} \leq T_\alpha$
and therefore $\bigwedge\{x_1,x_2\} \sqsubseteq_\alpha \bigwedge\{y_1,y_2\}$.

\vspace{0.1cm}
\noindent{Case $(\mnot\,\mathsf{E}_1)$:}
Assume $order(\lsem \mathsf{E}_1 \rsem_s(I)) = \alpha$. Then, by induction hypothesis
$\lsem \mathsf{E}_1 \rsem_s(I) \sqsubseteq_\alpha \lsem \mathsf{E}_1 \rsem_s(J)$ and thus
$order(\lsem \mathsf{E}_1 \rsem_s(J)) \geq \alpha$. It follows that $order(\lsem (\mnot\,\mathsf{E}_1)\rsem_s(I)) > \alpha$
and $order(\lsem (\mnot\,\mathsf{E}_1) \rsem_s(J)) > \alpha$ and therefore $\lsem (\mnot\,\mathsf{E}_1 )\rsem_s(I)
\sqsubseteq_\alpha \lsem (\mnot\,\mathsf{E}_1) \rsem_s(J)$.

\vspace{0.1cm}
\noindent{Case $(\exists\mathsf{V}. \mathsf{E}_1)$:}
Assume $\mathsf{V}$ is of type $\rho$.
It holds $\lsem (\exists\mathsf{V}. \mathsf{E}_1) \rsem_s(I) = \bigvee_{d \in \lsem \rho \rsem_{U_{\mathsf{P}}}}{\lsem \mathsf{E}_1 \rsem_{s[\mathsf{V}/d]}(I)}$.
It suffices to show $\bigvee_{d \in \lsem \rho \rsem_{U_{\mathsf{P}}}}{\lsem \mathsf{E}_1 \rsem_{s[\mathsf{V}/d]}(I)} \sqsubseteq_\alpha \bigvee_{d \in \lsem \rho \rsem_{U_{\mathsf{P}}}}{\lsem \mathsf{E}_1 \rsem_{s[\mathsf{V}/d]}(J)}$ which holds by induction hypothesis and Axiom~\ref{axiom6}.
\end{proof}

\def\modelintersect{\ref{model-intersection}}
\section{Proof of Theorem~\modelintersect{}}
We start by providing some necessary background material from~\cite{tocl14} on
how the $\bigsqcap$ operation on a set of interpretations is actually defined.


Let $x \in V$.
For every $X \subseteq (x]_\alpha$ we define $\bigsqcap_\alpha X$ as follows:
if $X = \emptyset$ then $\bigsqcap_\alpha X = T_\alpha$, otherwise
\[
\bigsqcap\nolimits_\alpha X =
                      \begin{cases}
                        \bigwedge X & order(\bigwedge X) \leq \alpha \\
                        T_{\alpha + 1} & \mbox{otherwise}
                      \end{cases}
\]
Let $\mathsf{P}$ be a program, $I \in {\cal I}_{\mathsf{P}}$ be a Herbrand interpretation of
$\mathsf{P}$ and $X \subseteq (I]_\alpha$.
For all predicate constants $\mathsf{p}$ in $\mathsf{P}$ of type $\rho_1 \to \cdots \to \rho_n \to o$ and
$d_i \in \lsem \rho_i \rsem_{U_{\mathsf{P}}}$ and for all $i = \{ 1, \ldots, n \}$,
it holds $\bigsqcap_\alpha X$ as $(\bigsqcap_\alpha X)(\mathsf{p})\ d_1\cdots\ d_n = \bigsqcap_\alpha\{ I(\mathsf{p})\ d_1\cdots\ d_n : I \in X\}$.

Let $X$ be a nonempty set of Herbrand interpretations. By Lemma~\ref{interpretations-lattice} we have
that ${\cal I}_\mathsf{P}$ is a complete lattice with respect to $\leq$ and a basic model.
Moreover, by Lemma~\ref{basic-model-lattice} it follows that ${\cal I}_\mathsf{P}$ is also a complete
lattice with respect to $\sqsubseteq$. Thus, there exist the least upper bound and greatest lower
bound of $X$ for both $\leq$ and $\sqsubseteq$. We denote the greatest lower bound of $X$ as $\bigwedge X$ and $\bigsqcap X$
with respect to relations $\leq$ and $\sqsubseteq$ respectively.
Then,  $\bigsqcap X$ can be constructed in an symmetric way to the least upper bound construction
described in \cite{tocl14}. More specifically, for each ordinal $\alpha < \kappa$ we define
the sets $X_\alpha, Y_\alpha \subseteq X$ and $x_\alpha \in {\cal I}_\mathsf{P}$, which are then
used in order to obtain $\bigsqcap X$.

Let $Y_0 = X$ and $x_0=\bigsqcap_0 Y_0$. For every $\alpha$, with $0< \alpha < \kappa$ we define
$X_\alpha = \{ x \in X : \forall\beta \leq \alpha\ x =_\alpha x_\alpha \}$,
$Y_\alpha = \bigcap_{\beta < \alpha}{X_\beta}$; moreover,
$x_\alpha = \bigsqcap_\alpha Y_\alpha$ if $Y_\alpha$ is nonempty and
$x_\alpha = \bigwedge_{\beta < \alpha} x_\beta$ if $Y_\alpha$ is empty.

Finally, we define $x_\infty = \bigwedge_{\alpha < \kappa}{x_\alpha}$.
In analogy to the proof of \cite{tocl14} for the least upper bound it can be shown
that $x_\infty = \bigsqcap X$ with respect to the relation $\sqsubseteq$.
Moreover, it is easy to prove that by construction it holds
$x_\alpha =_\alpha x_\beta$ and $x_\beta \geq x_\alpha$ for all $\beta < \alpha$.

\begin{lemma}\label{sqcap-is-model}
Let $\mathsf{P}$ be a program, $\alpha < \kappa$ and $M_\alpha$ be a Herbrand model of $\mathsf{P}$. Let ${\cal M} \subseteq (M_\alpha]_\alpha$ be a nonempty set of Herbrand models of $\mathsf{P}$.
Then, $\bigsqcap_\alpha{\cal M}$ is also a Herbrand model of $\mathsf{P}$.
\end{lemma}
\begin{proof}
Assume $\bigsqcap_\alpha{\cal M}$ is not a model.
Then, there exists a clause $\mathsf{p} \leftarrow \mathsf{E}$ in $\mathsf{P}$ and
$d_i \in \lsem \rho_i \rsem_D$ such that
$\lsem \mathsf{E} \rsem(\bigsqcap_\alpha{\cal M})\,d_1\cdots d_n > (\bigsqcap_\alpha{\cal M})(\mathsf{p})\,d_1\cdots d_n$. Since for every $N \in {\cal M}$ we have
$\bigsqcap_\alpha{\cal M} \sqsubseteq_\alpha N$, using Lemma~\ref{monotonicity-of-semantics} we conclude
 $\lsem \mathsf{E} \rsem(\bigsqcap_\alpha{\cal M}) \sqsubseteq_\alpha \lsem \mathsf{E} \rsem(N)$.
Let $x = \bigsqcap_\alpha\{ N(\mathsf{p})\,d_1\cdots d_n : N \in {\cal M}\}$. By definition,
 $x = (\bigsqcap_\alpha{\cal M})(\mathsf{p})\,d_1\cdots d_n$.

If $order(x) = \alpha$ then $x = \bigwedge\{ N(\mathsf{p})\,d_1\cdots d_n : N \in {\cal M}\}$.
If $x = T_\alpha$ then for all $N \in {\cal M}$ we have $N(\mathsf{p})\,d_1\cdots d_n = T_\alpha$.
Moreover, $\lsem \mathsf{E} \rsem(\bigsqcap_{\alpha}{\cal M})\,d_1\cdots d_n > T_\alpha$ and
by $\alpha$-monotonicity we have $\lsem \mathsf{E} \rsem(N)\,d_1\cdots d_n > T_\alpha$ for all $N \in {\cal M}$.
Then, $N(\mathsf{p})\,d_1\cdots d_n < \lsem \mathsf{E} \rsem(N)\,d_1\cdots d_n$ and
therefore $N$ is not a model (contradiction). If $x = F_\alpha$ then there exists $N \in {\cal M}$
such that $N(\mathsf{p})\,d_1\cdots d_n = F_\alpha$ and since $N$ is a model we have
$\lsem \mathsf{E} \rsem(N)\,d_1\cdots d_n \leq F_\alpha$. But then, it follows
$\lsem \mathsf{E} \rsem(\bigsqcap_\alpha{\cal M})\,d_1\cdots d_n \leq F_\alpha$ and
$\lsem \mathsf{E} \rsem(\bigsqcap_\alpha{\cal M})\,d_1\cdots d_n \leq x$ (contradiction).

If $order(x) < \alpha$ then $x = M_\alpha(\mathsf{p})\,d_1\cdots d_n$. If $x = T_\beta$
then $\lsem \mathsf{E} \rsem(\bigsqcap_\alpha {\cal M})\,d_1\cdots d_n > T_\beta$ and
$\lsem \mathsf{E} \rsem(M_\alpha)\,d_1\cdots d_n > T_\beta$. Then, we have $M_\alpha(\mathsf{p})\,d_1\cdots d_n <
\lsem \mathsf{E} \rsem(M_\alpha)$ and thus $M_\alpha$ is not a model of $\mathsf{P}$ (contradiction). If $x = F_\beta$ then
$\lsem \mathsf{E} \rsem(M_\alpha)\,d_1\cdots d_n \leq F_\beta$
and by $\alpha$-monotonicity
$\lsem \mathsf{E} \rsem(\bigsqcap_\alpha{\cal M})\,d_1\cdots d_n \leq F_\beta$.
Therefore, $\lsem \mathsf{E} \rsem(\bigsqcap_\alpha{\cal M})\,d_1\cdots d_n \leq x$ (contradiction).

If $order(x) > \alpha$ then $x = T_{\alpha+1}$ and
there exists model $N \in {\cal M}$ such that $N(\mathsf{p})\,d_1\cdots d_n < T_\alpha$.
Moreover, we have $\lsem \mathsf{E} \rsem(\bigsqcap_\alpha{\cal M})\,d_1\cdots d_n \geq T_\alpha$
and by $\alpha$-monotonicity we conclude $\lsem \mathsf{E} \rsem(N)\,d_1\cdots d_n \geq T_\alpha$.
But then, $\lsem \mathsf{E} \rsem(N)\,d_1\cdots d_n > N(\mathsf{p})\,d_1\cdots d_n$
and therefore $N$ is not a model of $\mathsf{P}$ (contradiction).
\end{proof}

%

In the following, we will make use of the following lemma that has been shown in~\cite[Lemma 3.18]{tocl14}:
\begin{lemma}\label{lub-is-equal-alpha}
If $\alpha \leq \kappa$ is an ordinal and $(x_\beta)_{\beta < \alpha}$ is a sequence
of elements of $L$ such that $x_\beta =_\beta x_\gamma$ and $x_\beta \leq x_\gamma$
($x_\beta \geq x_\gamma$) whenever $\beta < \gamma < \alpha$, and if $x = \bigvee_{\beta < \alpha}{x_\beta}$ ($x = \bigwedge_{\beta < \alpha}{x_\beta}$), then $x_\beta =_\beta x$ holds for all
$\beta < \alpha$.
\end{lemma}

\begin{lemma}\label{bigwedge-is-model}
Let $(M_\alpha)_{\alpha < \kappa}$ be a sequence of Herbrand models of $\mathsf{P}$ such that
$M_\alpha =_\alpha  M_\beta$ and $M_\beta \leq M_\alpha$ for all $\alpha < \beta < \kappa$.
Then, $\bigwedge_{\alpha < \kappa}{M_\alpha}$ is also a Herbrand model of $\mathsf{P}$.
\end{lemma}
\begin{proof}
Let $M_\infty = \bigwedge_{\alpha < \kappa}{M_\alpha}$ and assume $M_\infty$ is not a model of $\mathsf{P}$. Then, there is a clause
$\mathsf{p} \leftarrow \mathsf{E}$ and $d_i \in \lsem \rho_i \rsem_D$ such that
$\lsem \mathsf{E} \rsem(M_\infty)\,d_1\cdots d_n > M_\infty(\mathsf{p})\,d_1\cdots d_n$.
We define  $x_\alpha = M_\alpha(\mathsf{p})\,d_1\cdots d_n$, $x_\infty = M_\infty(\mathsf{p})\,d_1\cdots d_n$,
$y_\alpha = \lsem \mathsf{E} \rsem(M_\alpha)\,d_1\cdots d_n$ and
$y_\infty = \lsem \mathsf{E} \rsem(M_\infty)\,d_1\cdots d_n$ for all $\alpha < \kappa$.
It follows from Lemma~\ref{lub-is-equal-alpha} that  $M_\infty =_\alpha M_\alpha$
and thus $x_\infty =_\alpha x_\alpha$ for all $\alpha < \kappa$. Moreover, using $\alpha$-monotonicity
we also have $\lsem \mathsf{E} \rsem(M_\infty) =_\alpha \lsem \mathsf{E} \rsem(M_\alpha)$
and thus $y_\infty =_\alpha y_\alpha$ for all $\alpha < \kappa$.
We distinguish cases based on the value of $x_\infty$.

Assume $x_\infty = T_\delta$ for some $\delta < \kappa$. It follows by assumption
that $y_\infty > T_\delta$.
Then, since $x_\infty =_\delta x_\delta$ it follows $x_\delta = T_\delta$. Moreover,
since $y_\infty =_\delta y_\delta$ and $order(y_\infty) < \delta$ it follows $y_\delta = y_\infty > T_\delta$.
But then, $y_\delta > x_\delta$ (contradiction since $M_\delta$ is a model by assumption).

Assume $x_\infty = F_\delta$ for some $\delta < \kappa$. Then, since $x_\infty =_\delta x_\delta$
it follows $x_\delta = F_\delta$. Then, since $M_\delta$ is a model it follows $y_\delta \leq x_\delta$
and thus $y_\delta \leq F_\delta$. But then, since $y_\infty =_\delta y_\delta$ it follows
$y_\delta = y_\infty \leq F_\delta$. Therefore, $y_\infty \leq x_\infty$ that is a contradiction
to our assumption that $y_\infty > x_\infty$.

Assume $x_\infty = 0$. Then, $y_\infty > x_\infty = 0$. Let $y_\infty = T_\beta$ for some $\beta < \kappa$.
Then, since $y_\beta =_\beta y_\infty$ it follows $y_\beta = T_\beta$.  Since $M_\beta$ is a model of $\mathsf{P}$
it holds $T_\beta = y_\beta \leq x_\beta$, that is $x_\beta = T_\gamma$ for some $\gamma \leq \beta$. Moreover,
since $x_\infty =_\beta x_\beta$ it follows that $x_\infty = T_\gamma$ that is a contradiction to our assumption
that $x_\infty = 0$.
\end{proof}

\begingroup
\def\thetheorem{\ref{model-intersection}}
\begin{theorem}[Model Intersection Theorem]
Let $\mathsf{P}$ be a program and ${\cal M}$ be a nonempty set of
Herbrand models of $\mathsf{P}$. Then, $\bigsqcap{\cal M}$ is also
a Herbrand model of $\mathsf{P}$.
\end{theorem}
\addtocounter{theorem}{-1}
\endgroup
\begin{proof}
We use the construction for $\bigsqcap {\cal M}$ described in the beginning of this
appendix. More specifically, we define sets ${\cal M}_\alpha, Y_\alpha \subseteq {\cal M}$ and
$M_\alpha \in {\cal I}_{\mathsf{P}}$. Let $Y_0 = {\cal M}$ and
$M_0 = \bigsqcap_0 Y_0$. For every $\alpha >0$,
let ${\cal M}_\alpha = \{ M \in {\cal M} : \forall\beta \leq \alpha\ M =_\alpha M_\alpha \}$ and
$Y_\alpha = \bigcap_{\beta < \alpha}{{\cal M}_\beta}$; moreover,
$M_\alpha = \bigsqcap_\alpha Y_\alpha$ if $Y_\alpha$ is nonempty and
$M_\alpha = \bigwedge_{\beta < \alpha} M_\beta$ if $Y_\alpha$ is empty.
Then, $\bigsqcap{\cal M} = \bigwedge_{\alpha < \kappa}{M_\alpha}$.
It is easy to see that $M_\alpha =_\alpha M_\beta$ and ${\cal M}_\beta \supseteq {\cal M}_\alpha$
for all $\beta < \alpha$.

We distinguish two cases.
First, consider the case when $Y_\alpha$ is nonempty for all $\alpha < \kappa$. Then,
$M_\alpha = \bigsqcap_\alpha Y_\alpha$ and by Lemma~\ref{sqcap-is-model} it follows
that $M_\alpha$ is a model of $\mathsf{P}$. Moreover, by Lemma~\ref{bigwedge-is-model}
we get that $M_\infty = \bigwedge_{\alpha < \kappa}{M_\alpha}$ is also a model of $\mathsf{P}$.

Consider now the case that there exists a least ordinal $\delta < \kappa$ such that $Y_\delta$ is empty.
It holds (see~\cite{tocl14}) that $M_\infty = \bigwedge_{\alpha < \delta}{M_\delta}$. Suppose $M_\infty$ is not
a model of $\mathsf{P}$. Then, there is a clause
$\mathsf{p} \leftarrow \mathsf{E}$, a Herbrand state $s$ and $d_i \in \lsem \rho_i \rsem_D$ such that
$\lsem \mathsf{E} \rsem(M_\infty)\,d_1\cdots d_n > M_\infty(\mathsf{p})\,d_1\cdots d_n$.
We define  $x_\alpha = M_\alpha(\mathsf{p})\,d_1\cdots d_n$, $x_\infty = M_\infty(\mathsf{p})\,d_1\cdots d_n$,
$y_\alpha = \lsem \mathsf{E} \rsem(M_\alpha)\,d_1\cdots d_n$, and
$y_\infty = \lsem \mathsf{E} \rsem(M_\infty)\,d_1\cdots d_n$ for all $\beta \leq \alpha$.
We distinguish cases based on the value of $x_\infty$.

Assume $x_\infty = T_\beta$ for some $\beta < \delta$. It follows by assumption that $y_\infty > x_\infty = T_\beta$.
Then, by Lemma~\ref{lub-is-equal-alpha} it holds that $M_\infty =_\beta M_\beta$ and we
get $x_\infty =_\beta x_\beta$ and therefore $x_\beta = T_\beta$.
Moreover, by $\alpha$-monotonicity we
get $\lsem \mathsf{E} \rsem(M_\infty)\,d_1\cdots d_n =_\beta \lsem \mathsf{E} \rsem(M_\beta)\,d_1\cdots d_n$
and it follows that $y_\infty =_\beta y_\beta$. Moreover, since $y_\infty > T_\beta$ it follows $y_\beta = y_\infty > T_\beta$
and $y_\beta > x_\beta$. Since $Y_\beta$ is not empty by assumption we have that $M_\beta = \bigsqcap_\beta Y_\beta$
and by Lemma~\ref{sqcap-is-model} we get that $M_\beta$ is a model of $\mathsf{P}$ (contradiction since $y_\beta > x_\beta$).

Assume $x_\infty = F_\beta$ for some $\beta < \delta$. Then, by Lemma~\ref{lub-is-equal-alpha} it
holds $M_\infty =_\beta M_\beta$ and therefore $x_\infty =_\beta x_\beta$. It follows $x_\beta = F_\beta$.
Moreover, since $Y_\beta$ is nonempty by assumption and by Lemma~\ref{sqcap-is-model} it follows that
$M_\beta = \bigsqcap_\beta Y_\beta$ is a model of $\mathsf{P}$ and thus $y_\beta \leq x_\beta = F_\beta$.
By $\alpha$-monotonicity we get $\lsem \mathsf{E} \rsem(M_\infty) =_\beta \lsem \mathsf{E} \rsem(M_\beta)$
and therefore $y_\infty =_\beta y_\beta \leq F_\beta$. It follows $y_\infty \leq F_\beta = x_\infty$ (contradiction
to the initial assumption $y_\infty  > x_\infty$).

Assume $x_\infty = T_\delta$. By assumption we have $y_\infty > x_\infty = T_\delta$. Then, let $y_\infty = T_\gamma$
for some $\gamma < \delta$. By Lemma~\ref{lub-is-equal-alpha} it holds $M_\infty =_\gamma M_\gamma$ and by
$\alpha$-monotonicity it follows $\lsem \mathsf{E} \rsem(M_\infty) =_\gamma \lsem \mathsf{E} \rsem(M_\gamma)$
and thus $y_\infty =_\gamma y_\gamma$. It follows that $y_\gamma = T_\gamma$. Moreover, since $\gamma < \delta$
we know by assumption that $Y_\gamma$ is nonempty and therefore $M_\gamma = \bigsqcap Y_\gamma$ and by Lemma~\ref{sqcap-is-model}
$M_\gamma$ is a model of $\mathsf{P}$. It follows $T_\gamma = y_\gamma \leq x_\gamma$, that is, $x_\gamma = T_\beta$ for some
$\beta \leq \gamma < \delta$. Moreover, since $x_\infty =_\gamma x_\gamma$ it follows $x_\infty = T_\beta$ that is a contradiction
(since by assumption $x_\infty = T_\delta$).

Assume $x_\infty = F_\delta$. This case is not possible. Recall that $Y_\alpha$ is not empty for all $\alpha < \delta$
and thus $M_\alpha = \bigsqcap Y_\alpha$. By the definition of $\bigsqcap_\alpha$ we observe that either
$\bigsqcap_\alpha Y_\alpha \leq F_\alpha$ or $\bigsqcap_\alpha Y_\alpha \geq T_{\alpha + 1}$. Then, since $M_\infty = \bigwedge_{\alpha < \delta}{M_\alpha}$ it is not possible to have $x_\infty = F_\delta$.

Assume $x_\infty = 0$. This case does not arise. Again, $Y_\alpha$ is not empty for all $\alpha < \delta$ and thus
$M_\alpha = \bigsqcap_\alpha Y_\alpha$. Moreover, by definition of $\bigsqcap_\alpha$, $x_\alpha \neq 0$ for all $\alpha < \delta$.
Moreover, since $M_\infty = \bigwedge_{\alpha < \delta}{M_\alpha}$ and since $\delta < \kappa$ it follows that the
limit can be at most $T_\delta$.
\end{proof}

\def\leastmodel{\ref{least-model}}
\def\denotlemma{\ref{denotation-lemma}}
\def\tpmonotonic{\ref{tp-a-monotonic}}
\section{Proofs of Lemmas~\denotlemma{},~\tpmonotonic{} and Theorem~\leastmodel{}}

\begingroup
\def\thelemma{\ref{denotation-lemma}}
\begin{lemma}
Let $\mathsf{P}$ be a program.
For every predicate constant $\mathsf{p}:\pi$ in $\mathsf{P}$ and $I \in {\cal I}_\mathsf{P}$, $T_P(I)(\mathsf{p}) \in \lsem \pi \rsem_{U_\mathsf{P}}$.
\end{lemma}
\addtocounter{lemma}{-1}
\endgroup
\begin{proof}
It follows from the fact that $\lsem \pi \rsem_{U_\mathsf{P}}$ is a complete lattice (Lemma~\ref{pi-is-model}).
\end{proof}

\begingroup
\def\thelemma{\ref{tp-a-monotonic}}
\begin{lemma}
Let $\mathsf{P}$ be a program. Then, $T_\mathsf{P}$ is $\alpha$-monotonic for all $\alpha < \kappa$.
\end{lemma}
\addtocounter{lemma}{-1}
\endgroup
\begin{proof}
Follows directly from Lemma~\ref{monotonicity-of-semantics} and Proposition~\ref{lub-of-a-monotonic}.
\end{proof}

\begin{lemma}\label{TP-model-leq-model}
Let $\mathsf{P}$ be a program. Then, $M \in {\cal I}_\mathsf{P}$ is a model
of $\mathsf{P}$ if and only if $T_\mathsf{P}(M) \leq_{{\cal I}_\mathsf{P}} M$.
\end{lemma}
\begin{proof}
An interpretation $I \in {\cal I}_\mathsf{P}$ is a model of $\mathsf{P}$
iff $\lsem \mathsf{E} \rsem(I) \leq_\pi I(\mathsf{p})$ for all
clauses $\mathsf{p} \leftarrow_\pi \mathsf{E}$ in $\mathsf{P}$ iff
$\bigvee_{(\mathsf{p}\leftarrow\mathsf{E}) \in \mathsf{P}}{\lsem \mathsf{E} \rsem(I)} \leq_{{\cal I}_\mathsf{P}} I(\mathsf{p})$ iff $T_\mathsf{P}(I) \leq_{{\cal I}_\mathsf{P}} I$.
\end{proof}

\begin{proposition}\label{axiom5-for-pi}
Let $D$ be a nonempty set, $\pi$ be a predicate type and $x,y \in \lsem \pi \rsem_D$.
If $x \leq_\pi y$ and $x =_\beta y$ for all $\beta < \alpha$ then $x \sqsubseteq_\alpha y$.
\end{proposition}
\begin{proof}
The proof is by structural induction on $\pi$.

\vspace{0.1cm}
\noindent{\em Induction Basis:}
If $x =_\beta y$ for all $\beta < \alpha$ then either $x = y$ or
$order(x),order(y) \geq \alpha$. If $x = y$ then $x \sqsubseteq_\alpha y$.
Suppose $x \neq y$. If $order(x), order(y) > \alpha$ then $x =_\alpha y$.
If $x = F_\alpha$ then clearly $x \sqsubseteq_\alpha y$.
If $x = T_\alpha$ then $T_\alpha \leq y$ and therefore $y = T_\alpha$.
The case analysis for $y$ is similar.

\vspace{0.1cm}
\noindent{\em Induction Step:}
Assume that the statement holds for $\pi$.
Let $f, g \in \lsem \rho \rightarrow \pi \rsem_D$ and $\alpha < \kappa$.
For all $x \in \lsem \rho \rsem_D$ and $\beta < \alpha$, $f(x) \leq g(x)$ and $f(x) =_\beta g(x)$.
It follows that $f(x) \sqsubseteq_\alpha g(x)$. Therefore, $f \sqsubseteq_\alpha g$.
\end{proof}

\begin{proposition}\label{axiom5-for-interpretations}
Let $\mathsf{P}$ be a program and $I, J$ be Herbrand interpretations of $P$.
If $I \leq_{{\cal I}_\mathsf{P}} J$ and $I =_\beta J$ for all $\beta < \alpha$ then $I \sqsubseteq_\alpha J$.
\end{proposition}
\begin{proof}
Let $I,J \in {\cal I}_\mathsf{P}$ and $\alpha < \kappa$.
For all predicate constants $\mathsf{p}$ and $\beta < \alpha$, $I(\mathsf{p}) \leq J(\mathsf{p})$ and $I(\mathsf{p}) =_\beta J(\mathsf{p})$.
It follows by Proposition~\ref{axiom5-for-pi} that $I(\mathsf{p}) \sqsubseteq_\alpha J(\mathsf{p})$
and therefore, $I \sqsubseteq_\alpha J$.
\end{proof}

\begin{lemma}\label{if-model-then-prefixpoint}
Let $\mathsf{P}$ be a program. If $M$ is a model of $\mathsf{P}$ then $T_\mathsf{P}(M) \sqsubseteq M$.
\end{lemma}
\begin{proof}
It follows from Lemma~\ref{TP-model-leq-model} that if $M$ is a Herbrand model of $\mathsf{P}$
then $T_\mathsf{P}(M) \leq_{{\cal I}_\mathsf{P}} M$. If $T_\mathsf{P}(M) = M$ then the statement is immediate.
Suppose $T_\mathsf{P}(M) <_{{\cal I}_\mathsf{P}} M$ and let $\alpha$ denote the least ordinal such
that $T_\mathsf{P}(M) =_\alpha M$ does not hold. Then, $T_\mathsf{P}(M) =_\beta M$ for all $\beta < \alpha$.
Since $T_\mathsf{P}(M) <_{{\cal I}_\mathsf{P}} M$, by Proposition~\ref{axiom5-for-interpretations}
it follows that $T_\mathsf{P}(M) \sqsubseteq_\alpha M$. Since $T_\mathsf{P}(M) =_\alpha M$ does not hold,
it follows that $T_\mathsf{P}(M) \sqsubset_\alpha M$. Therefore $T_\mathsf{P}(M) \sqsubseteq M$.
\end{proof}

\begingroup
\def\thetheorem{\ref{least-model}}
\begin{theorem}[Least Fixed Point Theorem]
Let $\mathsf{P}$ be a program and let ${\cal M}$ be the set of all its Herbrand models.
Then, $T_\mathsf{P}$ has a least fixed point $M_\mathsf{P}$. Moreover, $M_\mathsf{P} = \bigsqcap{\cal{M}}$.
\end{theorem}
\addtocounter{theorem}{-1}
\endgroup
\begin{proof}
It follows from Lemma~\ref{tp-a-monotonic} and Theorem~\ref{a-monotonic-has-least-fixpoint}
that $T_\mathsf{P}$ has a least pre-fixed point with respect to $\sqsubseteq$ that is also a least fixed point.
Let $M_\mathsf{P}$ be that least fixed point of $T_\mathsf{P}$, i.e., $T_\mathsf{P}(M_\mathsf{P}) = M_\mathsf{P}$.
It is clear from Lemma~\ref{TP-model-leq-model} that $M_\mathsf{P}$ is a model of $\mathsf{P}$, i.e.,
$M_\mathsf{P} \in {\cal M}$. Then, it follows $\bigsqcap{\cal{M}} \sqsubseteq M_\mathsf{P}$. Moreover,
from Theorem~\ref{model-intersection} it is implied that $\bigsqcap{\cal{M}}$ is a model and thus from
Lemma~\ref{if-model-then-prefixpoint}, $\bigsqcap{\cal{M}}$ is a pre-fixed point of
$T_\mathsf{P}$ with respect to $\sqsubseteq$. Since $M_\mathsf{P}$ is the least pre-fixed point of $\mathsf{P}$,
$M_\mathsf{P} \sqsubseteq \bigsqcap{\cal M}$
and thus $M_\mathsf{P} = \bigsqcap{\cal M}$.
\end{proof}

\end{document}